\documentclass[sts]{imsart}

\usepackage{amsmath,amsthm,amssymb,amsopn,amsfonts,pdfpages,bm}
\usepackage{algpseudocode,algorithm}
\usepackage{fullpage}
\usepackage{parskip}
\usepackage{enumerate,url}
\usepackage{graphicx}
\usepackage{subfig}
\usepackage{graphics,multicol}
\usepackage{fancybox}
\usepackage{fancyhdr}
\usepackage[most]{tcolorbox}
\usepackage{MnSymbol}
\usepackage[colorlinks=true,pagebackref,linkcolor=magenta,citecolor=red]{hyperref} 
\usepackage{rotfloat}
\usepackage{booktabs}
\usepackage{multirow}
\usepackage{listings}
\makeatletter
\let \@fnsymbol\@arabic
\makeatother

\usepackage{todonotes}

\newtheorem{theorem}{Theorem}
\newtheorem{lemma}{Lemma}

\theoremstyle{definition}
\newtheorem{definition}{Definition}
\newtheorem{remark}{Remark}

\newcommand{\R}{\mathbb{R}}

\newcommand{\E}{\mathbb{E}}
\newcommand{\bA}{\mathbf{A}}
\newcommand{\bB}{\mathbf{B}}

\newcommand{\bH}{\mathbf{H}}

\newcommand{\bP}{\mathbf{P}}
\newcommand{\bR}{\mathbf{R}}
\newcommand{\bS}{\mathbf{S}}

\newcommand{\bW}{\mathbf{W}}
\newcommand{\bh}{\mathbf{h}}
\newcommand{\bx}{\mathbf{x}}
\newcommand{\bX}{\mathbf{X}}

\newcommand{\bC}{\mathcal{C}}

\newcommand{\bv}{\mathbf{v}}

\newcommand{\iid}{\stackrel{iid}{\sim}}

\newcommand{\Xhat}{\hat{\bX}}

\newcommand{\supp}{\operatorname{supp}}

\newcommand{\Wn}{\bW_n}

\newcommand{\bSigma}{\mathbf{\Sigma}}

\newcommand{\RDPG}{\operatorname{RDPG}}

\begin{document}

	\begin{frontmatter}
		
		\title{On estimation and inference in latent structure random graphs}
		\runtitle{Latent Structure Model}
		\author{\fnms{Avanti} \snm{Athreya}\corref{}\ead[label=e1]{dathrey1@jhu.edu}},
		\author{\fnms{Minh} \snm{Tang}\corref{}\ead[label=e2]{mtang8@ncsu.edu}},
		\author{\fnms{Youngser} \snm{Park}\corref{}\ead[label=e3]{youngser@jhu.edu}},
\author{\fnms{Carey E.} \snm{Priebe}\corref{}\ead[label=e4]{cep@jhu.edu}}
					\address{3400 N Charles Street, Baltimore, MD, 21218 and 2311 Stinson Drive, Raleigh, NC 27607\\
				 \printead{e1,e2,e3,e4}}
					\affiliation{Department of Applied Mathematics \& Statistics, Johns Hopkins University \\ Department of Statistics, North Carolina State University}
					\runauthor{Athreya et al.}
					\end{frontmatter}
\maketitle
	
%
%
%
%
	\section{Introduction}\label{sec:Intro}
 The last half-century has seen remarkable technical developments in random graph inference, the result of an integration across probabilistic combinatorics, classical statistics, and computer science. The ubiquity of graphs and networks in many applications, from urban planning to epidemiology to neuroscience, guarantees an enduring supply of real-world problems that rely on accurate graph inference for their resolution. Of course, a number of graph inference problems are comfortingly familiar and not necessarily peculiar to graphs per se: the parametric estimation of a common connection probability in an independent-edge random graph, for example, or the nonparametric estimation of a degree distribution.  Other inference tasks, such as community detection, are more graph-centric, and still others, such as vertex nomination \cite{FisLyzPaoChePri2015,Coppersmith2014} arise only in a network context. Nevertheless, even graph-specific inference tasks can frequently be resolved by appropriate Euclidean embeddings of graph data, and such Euclidean representations of graphs allow for a suite of classical statistical methods for Euclidean data, from estimation to classification to hypothesis testing \cite{athreya_survey}, to be effectively deployed in graph inference. 
 
 Advances in computational capacity now enable us to feasibly store and manipulate huge networks, but extracting from these data sets meaningful estimates and predictions, or inferring underlying relevant structure, remains a real challenge; at present, we are often confined to the realm of exploratory data analysis. We face, therefore, an ongoing need to synthesize the model-based inference procedures of twentieth-century statistics with the data-driven, algorithmically-propelled methods of twenty-first century machine learning. 
 In his landmark polemic on the two ``cultures," \cite{breiman_statsci}, Breiman described this very divide, and argued persuasively for the gains that machine learning can deliver. While we agree, we remain believers in a theoretical framework for graph inference that begins first with a compelling graph model. Such a model is useful not only because it allows us to generate, say, theoretical bounds for error rates in graph estimation procedures, but also because it offers a unifying perspective for graph analysis. 

  In this spirit, we present here the  {\em latent structure model} (LSM) for random graphs. We demonstrate that the LSM is tractable and useful, especially for inference tasks that involve the discovery or exploitation of lower-dimensional geometric structure. The LSM sits between two workhorse random graph models, the stochastic block model (SBM) \cite{Holland1983} and the random dot product graph model (RDPG) \cite{young2007random}.  That is, latent structure models impose parametric and geometric requirements on distributions that are more elaborate than those of a stochastic block model but more constrained than those of a typical random dot product graph.
   
 Latent structure random graphs are a special case of latent position random graphs \cite{hoff_raftery_handcock,diaconis2007graph,asta_cls}, which are a type of inhomogeneous Erd\H{o}s-R\'{e}nyi random graph \cite{bollobas2007phase} in which edges between any pairs of vertices arise independently of one another.  Every vertex in a latent position random graph has associated to it a (typically unobserved) {\em latent position}, itself an object belonging to some (often Euclidean) space $\mathcal{X}$.  Probabilities of an edge between two vertices $i$ and $j$, $p_{ij}$, are then a function $\kappa(\cdot,\cdot): \mathcal{X} \times \mathcal{X} \rightarrow [0,1]$ (known as the {\em link function}) of their associated latent positions $(x_i, x_j)$. Thus $p_{ij}=\kappa(x_i, x_j)$, and as mentioned previously, edges between vertices arise independently of one another. Given these probabilities, the entries $\bA_{ij}$ of the adjacency matrix $\bA$ are conditionally independent Bernoulli random variables with success probabilities $p_{ij}$. We consolidate these probabilities into a matrix $\bP=(p_{ij})$, and we write $\bA \sim \bP$ to denote this relationship.
 
 In a $d$-dimensional random dot product graph, the latent space is an appropriately-constrained subspace of $\mathbb{R}^d$, and the link function is simply the dot product of the two latent $d$-dimensional vectors.  A quintessential inference problem in an RDPG setting is the estimation of latent positions from a single observation of a suitably large graph. The linear algebraic foundation for an RDPG makes such an inference problem especially amenable to spectral methods, such as singular value decompositions, of adjacency or Laplacian matrices. Indeed, these spectral decompositions have been the basis for a suite of approaches to graph estimation, community detection, and hypothesis testing for random dot product graphs. (For a comprehensive summary of these techniques, see \cite{athreya_survey}.)
 Because of the invariance of the inner product to orthogonal transformations, however, the RDPG exhibits a clear nonidentifiability: latent positions can be estimated only up to an orthogonal transformation. Note that the popular stochastic blockmodel (SBM) can be regarded as a random dot product graph. In an SBM, there are a finite number of possible latent positions for each vertex---one for each block---and the latent position exactly determines the block assignment for that vertex. 
%
 
 Random dot product graphs are often divided into two types: those in which the latent positions are fixed, and those in which the latent positions are themselves random. Specifically, we consider the case in which the latent position $X_i \in \mathbb{R}^d$ for vertex $i$ is drawn from some distribution $F$ on $\mathbb{R}^d$, and we further assume that the latent positions for each vertex are drawn independently and identically from this distribution $F$. A common graph inference task is to infer properties of $F$ from an observation of the graph alone. For example, in a stochastic block model, in which the distribution $F$ is discretely supported, we may wish to estimate the point masses in the support of $F$. In the graph inference setting, however, there are two sources of randomness that culminate in the generation of the actual graph: first, the randomness in the latent positions, and second, {\em given these latent positions}, the conditional randomness in the existence of edges between vertices. As such, the task of inferring properties of the underlying distribution $F$ from a mere observation of the adjacency matrix $\bA$ is more complicated than the classical problem of inferring properties of $F$ directly from the $X_i$'s, the latter of which of course represent an i.i.d. sample from $F$. This is because {\em these latent positions $X_i$ are not observed in the first place}.  The key to such inference is the initial step of consistently estimating the unobserved $X_i$'s from $\bA$, and then using these estimates, denoted $\hat{X}_i$, to infer properties of $F$. 
 
 
 Now, an RDPG with i.i.d. latent positions allows for a wide range of possible distributions $F$, and by contrast, the SBM imposes the constraint of a discrete support for $F$. A natural midpoint between these two is to constrain $F$ to belong to a {\em parametric family} of distributions on some space $\mathcal{S}$: that is, $F \in \{F_{\theta}, \theta \in \mathbb{R}^l\}$, $\supp F \subset \mathcal{S}$. A useful example to keep in mind is $F \sim \textrm{Beta}(a,b)$, with $\supp F=[0,1]$, the unit interval, and $l$, the dimension of the parameter space, given by $l=2$. Here the latent positions are random points in the unit interval, so the associated RDPG has a one-dimensional latent space, and an inference task of interest is to estimate or test hypotheses about the parameters $a$ and $b$. We remark that the Beta distribution for latent positions provides a nice illustration of the fact that the dimension of the random dot product graph may be different than the number of unknown parameters.  In the case of the Beta distribution, the support $\mathcal{S}$ of $F$ is known, but in other cases, inferring the geometry of the support of the distribution may be part of our larger task.  
 
 
 Because constraints on $F$ impose additional structure---structure that can be both geometric, such as prescriptions on the parameter space or the support, and functional, such as limitations on the class of distributions themselves---we call graphs of this type {\em latent structure model (LSM)} graphs (see Def. \ref{def:LSM} in Sec. \ref{sec:Def_Note_Background}). Much of the rest of this manuscript is devoted to demonstrating that (a) statistical methodology for RDPGs can be successfully applied to yield estimates for model parameters and to conduct broader inference tasks in LSMs and (b) the structure within LSMs can be leveraged to obtain sharp rates of convergence for such estimates.
In short, the latent structure model is flexible, amenable to a suite of existing techniques for inference on RDPGs, and a useful starting point for models with more intricate geometric structure.
 
 
 As we have already emphasized, our approach to inference for a latent structure model is first to treat LSMs as RDPGs and use the considerable literature on the consistency and asymptotic normality of spectral estimates for latent positions in RDPGs \cite{lyzinski13:_perfec, lyzinski15_HSBM, athreya2013limit, STFP-2011,tang14:_nonpar}. More precisely, if the latent space dimension $d$ of an RDPG is known, \cite{lyzinski13:_perfec} and \cite{lyzinski15_HSBM} show that a rank $d$ singular value decomposition of the adjacency matrix $\bA$ gives a consistent estimate, denoted $\hat{\bX}$, for the matrix of latent positions $\bX$. In addition, \cite{athreya2013limit} demonstrates that as the number of vertices $n$ of the graph increases, the rows of $\hat{\bX}$ have an asymptotically normal distribution about the true latent positions. Further, \cite{tang14:_nonpar} establishes that the underlying distribution $F$ can be consistently recovered via kernel density estimation with these spectral estimates of the true latent positions. Most critically, \cite{tang14:_nonpar} ensures the convergence of an empirical process of the spectrally-estimated latent positions. This functional central limit theorem allows us to prove that in the latent structure model, when the latent position distribution belongs to a parametric family, one can effectively use these spectral estimates as ``data" to construct an M-estimate (essentially a quasi-maximum likelihood estimate) of the underlying parameter $\theta$, and, surprisingly, {\em still obtain a parametric rate of convergence} of such a quasi-MLE to its true value.  That the introduction of spectral estimates in place of the true latent positions does not change the asymptotic rate of convergence of this estimator is a testament to how accurate and valuable are the spectral estimates themselves, not only for recovering the true latent positions but for a variety of subsequent graph inference tasks.
 
 As an illustration of our result, we consider inference when the latent positions are distributed as points along the 1-dimensional Hardy-Weinberg curve in the simplex, defined as the image of $$r:[0,1] \rightarrow \mathbb{R}^3; r(t)=(t^2, 2 t (1-t), (1-t)^2)$$
 Let $p$ be the arclength reparameterization of this curve.
 Suppose that $t_i \in [0,1]$ are drawn independently from a common $G_{\theta}=\textrm{Beta}(\theta=(a, b))$ distribution, and consider an RDPG with latent positions $X_i=p(t_i)$ that lie on the Hardy-Weinberg curve. We note that the latent positions are points in the ambient space $\mathbb{R}^3$, which is the dimension of the resulting RDPG. But in fact, of course, the latent positions lie on the two-dimensional unit simplex $$(x_1, x_2, x_3): \sum_{i} x_i=1,\, \, 0 \leq x_i \leq 1$$
and more precisely still, they lie on the one-dimensional submanifold that is the Hardy-Weinberg curve. 

If we observe only the adjacency matrix $\bA$ for a random dot product graph with these latent positions $\bX$, how might we estimate or conduct tests about the parameters $a$ or $b$ of this underlying distribution $G$? One approach is to spectrally embed $\bA$ to obtain the point cloud of estimated latent positions (organized, as before, as rows of a matrix $\hat{\bX}$) in $\mathbb{R}^3$; rotate this point cloud appropriately (due to the nonidentifiability of the RDPG); project these estimated points onto the Hardy-Weinberg curve; pull these projected points back into the unit interval through $p^{-1}$, and use these projected, pulled-back points in the unit interval, denoted $\hat{Y}_i$, as ``data" in the estimation of the parameters of $G$. See Figure \ref{fig:HW1}, below, for a representation of the estimated latent positions of this LSM graph around the Hardy-Weinberg curve.
\begin{figure}[H]
	\centering
	\includegraphics[scale=0.5]{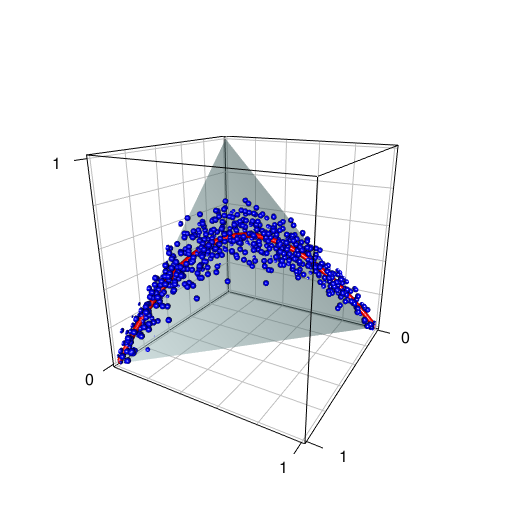}
	\caption{Estimated latent positions, with $n=1000$, in a tubular neighborhood about the Hardy-Weinberg curve when underlying distribution $G_{\theta}$ is $\textrm{Beta}(a=1, b=1)$.}
	\label{fig:HW1}
\end{figure} 
That is, we might plug these points $\hat{Y}_i$ in the unit interval---which, we stress, are neither independent nor identically distributed---into the estimating equations that define familiar maximum likelihood estimates for $(a, b)$. Though this procedure may be straightforward to write, it poses computational and mathematical pitfalls, even in the case when the geometric structure of the Hardy-Weinberg curve is known {\em a priori}. Also, while each step of this procedure is sensible, there are many sources of error. Given the cumulative impact of noise and dependence in the latent position estimates and bias from the projections and pullbacks, one might reasonably view this recipe as little more than a principled hack.  As it happens, however, in the case of a parametric latent structure model with known support, this quasi-$M$-estimation delivers both consistency and efficiency.

In more complicated latent structure models, the family of distributions $F$ may have support $\mathcal{S}$ that is unknown, and must also be inferred. This is a significantly more intricate problem, one for which our methodology currently admits fewer guarantees. To unify latent structure model estimation and inference over different levels of model complexity, we devote our next section to elucidating the different types of latent structure models, from models with known or parametric distributions on known support to models with nonparametrically-specified distributions over unknown support.

 
 We organize the paper as follows.  In Section \ref{sec:Def_Note_Background}, we define the latent structure model and relate it to stochastic block models and random dot product graphs. In Section \ref{sec:prior_results_RDPG}, we summarize key theoretical results for random dot product graphs, including consistency and normality, as well as a Donsker-class functional central limit theorem, for spectral estimates of latent positions. In Section \ref{sec:asymp_eff_LSM}, we demonstrate how these results can be exploited to give a parametric rate of convergence for estimates of LSM parameters. In Section \ref{sec:Examples}, we consider examples of estimation in specific latent structure models, including models with known support and models with parametric support that must be learned or estimated from the data. We conclude with an analysis of the right and left hemsipheres of the {\em Drosophila} larval connectome, which we view as a nonparametric latent structure model with unknown support. The framework of an LSM permits us to resolve, as a statistical test of hypothesis, the neuroscientific question of {\em bilateral homology}---that is, structural similarity across hemispheres---of the  {\em Drosophila} connectome. Finally, in Section \ref{sec:Conclusion}, we close with a discussion of the relevance of the LSM and associated open problems.
 
 \section{Definitions, notation, and background}\label{sec:Def_Note_Background}
 In our notation, we will use boldface $\bH$ to represent a matrix, and we use $H_i$ to represent the $i$th row of this matrix. We use $| \cdot |$ to represent Euclidean distance, with the dimension being clear from context. We use ${\top}$ to represent tranpose, and $\perp$ to denote the orthogonal complement. We use $P$ to denote probability and $\mathbb{E}$ to denote expectation.
 
 To begin, we define a {\em graph} $G$ to be an ordered pair of $(V,E)$ where $V$ is the {\em vertex} or {\em node} set, and $E$, the set of {\em edges}, is a subset of the Cartesian product of $V \times V$. In a graph whose vertex set has cardinality $n$, we will usually represent $V$ as $V=\{1, 2, \dots, n\}$, and we say there {\em is an edge between} $i$ and $j$ if $(i,j)\in E$.  The {\em adjacency} matrix $\bA$ provides a compact representation of such a graph:
 $$\bA_{ij}=1 \textrm{   if   }(i,j) \in E, \textrm{  and  }\bA_{ij}=0 \textrm{ otherwise. }$$ 
 Where there is no danger of confusion, we will often refer to a graph $G$ and its adjacency matrix $\bA$
 interchangeably.

 \subsection{Models}
 Since our focus is on latent structure models and we wish to exploit lower-dimensional geometric structure, we first clarify the notation of the smallest appropriate dimension for a latent structure model.
  \begin{definition}[Minimal subspace dimension]\label{def:min_subspace_dim}
  	Let $\tilde{\gamma}:[0,1] \rightarrow \mathbb{R}^k$ be a smooth (twice continuously differentiable) map and let $\bC=\textrm{Image}(\tilde{\gamma})$ be the curve that is the image of this map. We say that $\bC$ has {\em minimal subspace dimension} $d$, denoted $md(\bC)=d$, if 
  	$$\min\{\dim(S): S \subset \mathbb{R}^k \textrm{ a subspace}, \bC \subset S\}=d$$
  	\end{definition}
  	We stress that this linear subspace requirement is crucial---the Hardy-Weinberg curve lies in the simplex, which is a two-dimensional surface, but this plane does not pass through the origin; the simplex is not a linear subspace. Hence the minimum subspace dimension of the Hardy-Weinberg curve is 3, not 2.
  	
  	Next, since we frame our latent structure models as special cases of random dot product graphs, we define inner product distributions and inner product curves.
  	\begin{definition}
  		[$d$-dimensional inner product distribution and inner product curve]\label{def:innerprod}
  		Let $F$ be a probability distribution whose support is given by $\supp F={\bf \mathcal{X}}_d \subset \R^d$.
  		We say that $F$ is a
  		\emph{$d$-dimensional inner product distribution}
  		on $\R^d$ if for all $x,y \in \mathcal{X}_d=\supp F$, we have $x^{\top} y \in [0,1]$.
  		Next, let $\bC$ be a smooth (twice continuously differentiable) curve defined as $\bC=\textrm{Im}(\tilde{\gamma})$ where $\tilde{\gamma}:[0,1]\rightarrow \mathbb{R}^k$ is smooth. Suppose $md(C)=d$, and define $\gamma$ by
  		$$\gamma:[0,1] \rightarrow \mathbb{R}^d; \gamma=\pi_{k,\bC} \circ \tilde{\gamma}$$ 
  		where $\pi_{k, \bC}$ is the projection map from $\mathbb{R}^k$ onto $\bC$. We say that $\bC$ is a {\em non-self-intersecting, $d$-dimensional inner product curve} if (i) $\gamma$ is injective and has smooth inverse $\gamma^{-1}$ and (ii) for all $x, y \in C$, $x^{\top}y \in [0,1]$.
  	\end{definition}

The definition of inner product curves and distributions on suitable subsets of Euclidean space is a building block to the construction of a random dot product graphs and latent structure random graphs. We start with a random dot product graph, which we define as an independent-edge random graph
 for which the edge probabilities are given by the dot products of the latent
 positions associated to the vertices. The latent positions are necessarily constrained to have inner-product distributions.
 We restrict our attention here to graphs that are undirected and
 loop-free.
 \begin{definition} [Random dot product graphs \cite{young2007random}] \label{def:RDPG}
 	Let $F$ be a $d$-dimensional inner product distribution
 	with $X_1,X_2,\dots,X_n \iid F$, collected in the rows of the matrix
 	$$\bX=[X_1, X_2, \dots, X_n]^{\top} \in \R^{n \times d}.$$
    (Note that each $X_i$ is a column vector in $\mathbb{R}^d$, and in the matrix $\bX$, these column vectors are transposed and organized as rows.)
 	Suppose $\bA$ is a symmetric, hollow random adjacency matrix whose above diagonal entries are distributed as follows:
 	\begin{equation} \label{eq:rdpg}
 	P[\bA|\bX]=
 	\prod_{i<j}(X_i^{\top}X_j)^{\bA_{ij}}(1-X_i^{\top}X_j)^{1-\bA_{ij}}
 	\end{equation}
That is, conditional on the latent positions $\bX$, the above-diagonal entries $\bA_{ij}$ are independent Bernoulli random variables with $P(\bA_{ij} =1)=X_i^{\top} X_j$. To denote this, we write $(\bA,\bX) \sim \RDPG(F,n)$ and say that $\bA$ is the adjacency
 	matrix of a {\em random dot product graph (RDPG) of dimension or rank at most} $d$ and with {\em latent positions} given by the rows of $\bX$. If $\bX \bX^{\top}$ is, in fact, a rank $d$ matrix, we say $\bA$ is the adjacency matrix of a rank $d$ random dot product graph.
 	
 	If, instead, the latent positions are given by a fixed matrix $\bX$ and, given this matrix, the graph is generated according to Eq.\eqref{eq:rdpg}, we say that $\bA$ is a realization of a random dot product graph with latent positions $\bX$, and we write $\bA \sim \mathrm{RDPG}(\bX)$.
 	
 	Finally, let $\rho_n$ be a sequence of positive real numbers less than one, and suppose
		 	$X_1,\dotsc, X_n {\sim} F$ be independent random variables with $F$ an inner-product distribution. 
 		 	We say that  $(\bX,\bA)\sim \mathrm{RDPG}(F)$ {\em with sparsity factor} $\rho_n$ if $\bA$ is symmetric, hollow and consists of independent above-diagonal entries $\bA_{ij}$ distributed as $
 		  	\bA_{ij} \sim\mathrm{Bernoulli}(\rho_n X_i^\top X_j)$.
 \end{definition}
%
 
 \begin{remark} [Nonidentifiability]\label{rem:nonid}
 	Given a graph distributed as an RDPG,
 	the natural task is to recover the latent positions $\bX$ that gave
 	rise to the observed graph.
 	However, the RDPG model has an inherent nonidentifiability:
 	let $\bX \in \R^{n \times d}$ be a matrix of latent positions
 	and let $\bW \in \R^{d \times d}$ be a unitary matrix.
 	Since $\bX \bX^{\top} = (\bX \bW) (\bX \bW)^{\top}$, it is clear that the latent positions
 	$\bX$ and $\bX\bW$ give rise to the same distribution over graphs in
 	Eq.~\eqref{eq:rdpg}.
 	Note that most latent position models, as defined below, also suffer from similar types of non-identifiability as edge-probabilities may be invariant to various transformations.
 \end{remark}

 Random dot product graphs are special cases of more general {\em latent-position random graphs}, which are independent-edge random graphs in which each vertex has a latent position and for which connection probabilities are given by appropriate functions of these latent positions.  
  Conversely, while latent position models generalize the random dot product graph, RDPGs, in turn, are a generalization of the more limited {\em stochastic blockmodel} (SBM) graph \cite{Holland1983} and its variants such as the degree-corrected SBM \cite{karrer2011stochastic} and the mixed membership SBM \cite{Airoldi2008}. The stochastic block model is an independent-edge random graph whose vertex set is partitioned into $K$ groups, called {\em blocks}, and the stochastic blockmodel is typically parameterized by (1) a
 $K\times K$ matrix of probabilities $\bB$ of adjacencies between vertices in
 each of the blocks, and (2) a {\em block-assignment vector} $\tau:[n] \rightarrow [K]$ which assigns each vertex to its block. That is, for any two vertices $i,j$, the probability of their connection is 
 $$\bP_{ij}=\bB_{\tau(i), \tau(j)},$$
 and we typically write $\bA \sim \mathrm{SBM}(\bB, \tau)$.
 Here we present an alternative definition in terms of the RDPG model.
 
 \begin{definition}[Positive semidefinite $k$-block SBM]\label{def:PS_SBM} We say an RDPG with latent positions $\bX$ is an SBM with $K$ blocks if the
 	number of distinct rows in $\bX$ is $K$, denoted $\bX_{(1)}, \dots, \bX_{(K)}$  In this case, we define the
 	block membership function $\tau:[n]\mapsto [K]$ to be a function
 	such that $\tau(i)=\tau(j)$ if and only if $\bX_i=\bX_j$.  
 	We then write $$\bA \sim \mathrm{SBM}(\tau, \{\bX_{(i)}\}_{i=1}^{K})$$
 	In addition, we also consider the case of a stochastic block model in which the block membership of each vertex is randomly assigned. More precisely, let $\pi \in (0,1)^{K}$ with $\sum_{k=1}^{n} \pi_k=1$ and suppose that $\tau(1), \tau(2), \dots, \tau(n)$ are now i.i.d. random variables with distribution $\mathrm{Categorical}(\pi)$; that is, $\mathrm{Pr}(\tau(i) = k) = \pi_k$ for all $k$. Then we say $\bA$ is an {\em SBM with i.i.d block memberships}, and we write $$\bA \sim \mathrm{SBM}(\pi, \{X_{(i)}\}).$$
 \end{definition}
 
 With RDPGs and SBMs defined, we now define {\em latent structure random graphs} or {\em latent structure models} (LSMs) as, in effect, random dot product graphs of dimension $d$ whose latent position distributions are determined by a family of distributions on some appropriate, potentially lower-dimensional submanifold, which we call the {\em support} $\mathcal{S}$ of the distribution. Our definition begins with the simplest such models, in which the support $\mathcal{S}$ of $F$ is known and the knowledge of the parameters uniquely identifies the distribution within a family, to increasingly more complex cases in which the support of the latent position distribution may itself be unknown. Latent structure models have two critical components: one, a known or estimable curve or manifold, the {\em structural support}, on which the latent position distribution $F$ is supported; and two, a further so-called {\em underlying} distribution $G$ on some other fixed subset of Euclidean space (in our one-dimensional setting, this is the unit interval). Therefore, they naturally bifurcate along these two axes: first, whether the structural support is known, can be constrained to belong to a certain family of submanifolds, or is (mostly) unconstrained;  second, whether the underlying distribution $G$ in Euclidean space is known, parametrically specified, or nonparametric. We consolidate these hierarchical notions in Definition \ref{def:LSM} below.
 
For simplicity and clarity, in this paper we define and focus on one-dimensional latent structure models, in which the structural support $\mathcal{S}$ is a curve $\mathcal{C}$. Given a finite length inner product curve $\bC$ with minimal subspace dimension $d$, let $T_{\bC}(R)$ be a tubular neighborhood (see \cite{lee_manifold}) of radius $R$ about $\bC$. To avoid pathologies, we restrict ourselves to structural support curves that satisfy certain regularity conditions.
\begin{definition} A smooth, finite length inner product curve $\mathcal{C}$ of minimal subspace dimension $d$ is said to be an {\em LSM-regular structural support curve} if there exists a tubular neighborhood of positive radius $R>0$ about $\mathcal{C}$ on which the projection map $\pi_C: \mathbb{R}^d \mapsto \bC$ satisfying $\pi_C(x)=\textrm{argmin}_{y \in \bC} |x-y|$ is well-defined and twice-continuously differentiable.
\end{definition}
\begin{definition}[One-dimensional latent structure model]\label{def:LSM}
 Let $\bC$ be an LSM-regular curve of minimal subspace dimension $d$. Let $p(t):[0,1] \rightarrow \bC$ denote the arclength reparameterization of $\bC$.  Let $\mathcal{G}$ be a family of distributions $\{G: G \in \mathcal{G}\}$ on $[0,1]$ with associated distribution measures $\{\mu_G: G \in \mathcal{G}\}$. Let $\mathcal{F}$ denote the family of associated induced distributions on $\bC$; that is, for each $F$ in $\mathcal{F}$, the distribution measure $\mu_F$ is given by $\mu_F(B)=\mu_G(p^{-1}(B))$ for any Borel set $B$.
 We say that an RDPG with i.i.d latent position matrix $\bX$ is a {\em parametric latent structure random graph with known univariate support $\bC$ and underlying distribution $G$} if the latent position vectors $X_i$ are distributed according to $F=G(p^{-1})$ where $G$ belongs to some regular parametric family $\mathcal{G}_{\Theta}=\{G_{\theta}; \theta \in \Theta \subset \mathbb{R}^l\}$ on $[0,1]$ and $p$ and $\bC$ are known.  We write
 	$$X_i \sim F=G_{\theta}(p^{-1}),\theta \in \Theta;\, \,  \supp F=\bC$$
 	We say that an RDPG with iid latent position matrix $\bX$ is a {\em nonparametric latent structure random graph with known univariate support $\bC$} if  $p$ and $\bC$ are both known, and $F=G(p^{-1})$, where $G \in \mathcal{G}$ with $\mathcal{G}$ a family of distributions on $[0,1]$ that is not a subset of any regular parametric family of distributions on $[0,1]$.
 
 	Next, we say that an RDPG with i.i.d latent position matrix $\bX$ is a {\em parametric latent structure random graph with parametrically determined univariate support and underlying distribution $G$} if, first, the rows $X_i$ of $\bX$ are given by the distribution $F$ on $\bC$, where $F=G(p^{-1})$ and $G$ belongs to a parametric family of distributions $\{G_{\theta}: \theta \in \Theta \subset \mathbb{R}^l\}$ on $[0,1]$; and second, the map $p:[0,1]\rightarrow \bC$ is uniquely determined (up to orientation; see Remark \ref{rem:orientation_nonid}) by a vector $\eta \in \mathbb{R}^q$. We say that an RDPG with iid latent position matrix $\bX$ is a {\em nonparametric latent structure random graph with parametrically determined univariate support} if  $p:[0,1]\rightarrow \bC$ is uniquely determined (up to orientation) by a vector $\eta \in \mathbb{R}^q$ and $F=G(p^{-1})$, where $G \in \mathcal{G}$ with $\mathcal{G}$ a family of distributions on $[0,1]$ that is not a subset of any regular parametric family of distributions on $[0,1]$.
 	
 	Finally, we say that an RDPG with i.i.d latent position matrix $\bX$ is a {\em parametric latent structure random graph with nonparametric univariate support and underlying distribution $G$} if the rows $X_i$ of $\bX$ are given by distribution $F=G(p^{-1})$, where $G$ belongs to a parametric family of distributions $\{G_{\theta}: \theta \in \Theta \subset \mathbb{R}^l\}$ on $[0,1]$ and $p$ is not constrained to be uniquely determined (up to orientation) by a fixed vector $\eta \in \mathbb{R}^q$. We say that an RDPG with i.i.d latent positions matrix $\bX$ is a {\em nonparametric latent structure random graph with nonparametric univariate support} if the rows $X_i$ of $\bX$ are given by distribution $F=G(p^{-1})$, where $G \in \mathcal{G}$ with $\mathcal{G}$ a family of distributions on $[0,1]$ that is not a subset of any regular parametric family of distributions on $[0,1]$, and $p$ is not constrained to be uniquely determined up to orientation by any fixed, finite-dimensional vector. 
 	\end{definition}
 	\begin{remark}[Arclength and nonidentifiability up to orientation]\label{rem:orientation_nonid}
Because $F$ is defined as an induced distribution on the curve $\mathcal{C}$, the specification of the arclength parameterization $p$ is necessary to avoid nonidentifiability. We remark that the arclength parameterization is unique up to orientation; that is, up to the transformation $t \mapsto 1-t$. Thus, our latent structure models are identifiable only up to an orientation. Hence latent structure models have two distinct sources of nonidentifiability. The first is a nonidentifiability inherited directly from the random dot product graph, namely invariance of the inner product to orthogonal transformation.  The second is the parametrization nonidentifiability that governs how the map $p:[0,1] \mapsto \mathcal{C}$ is written, or equivalently, the location of $p(0)$ on $\mathcal{C}$.
 	\end{remark} 
Figure \ref{fig:HW1mix} depicts precisely such a latent structure model.  Here, the underlying distribution $G$ on $[0,1]$ is a mixture of two Beta distributions, shown in panel (a), and  the curve is the Hardy-Weinberg curve. On this curve, in panel (b), we see the distribution of the latent positions on the Hardy-Weinberg curve; this is a representation, on the Hardy-Weinberg curve, of the transformed mixture of Beta densities in the unit interval. Panel (c) of Fig. \ref{fig:HW1mix} shows the random dot product graph generated from and i.i.d sample of latent positions on the Hardy-Weinberg curve.
\begin{figure}
	\subfloat[]{\includegraphics[width=0.25\textwidth]{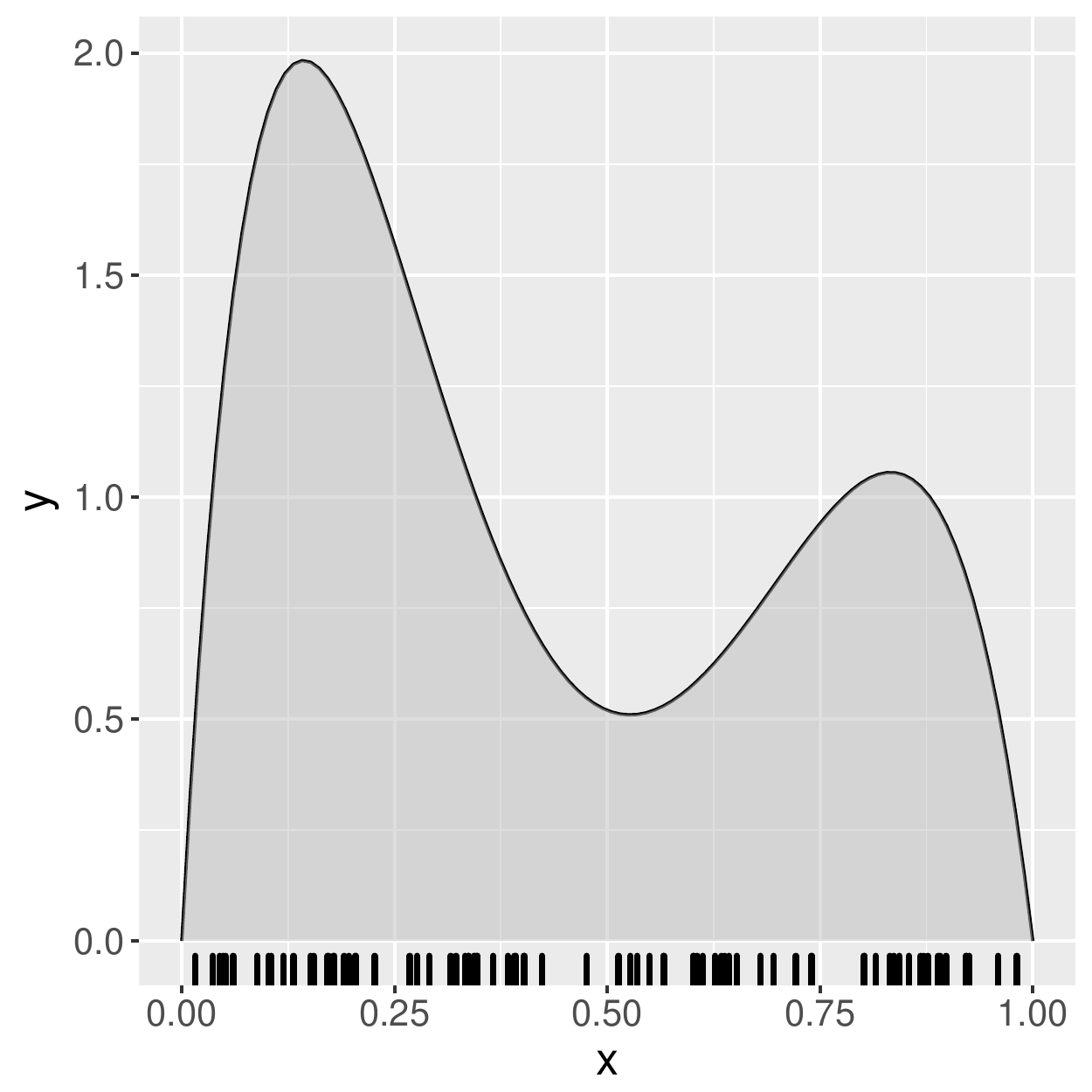}}
\hspace{0.4in}
\subfloat[]{\includegraphics[width=0.25\textwidth]{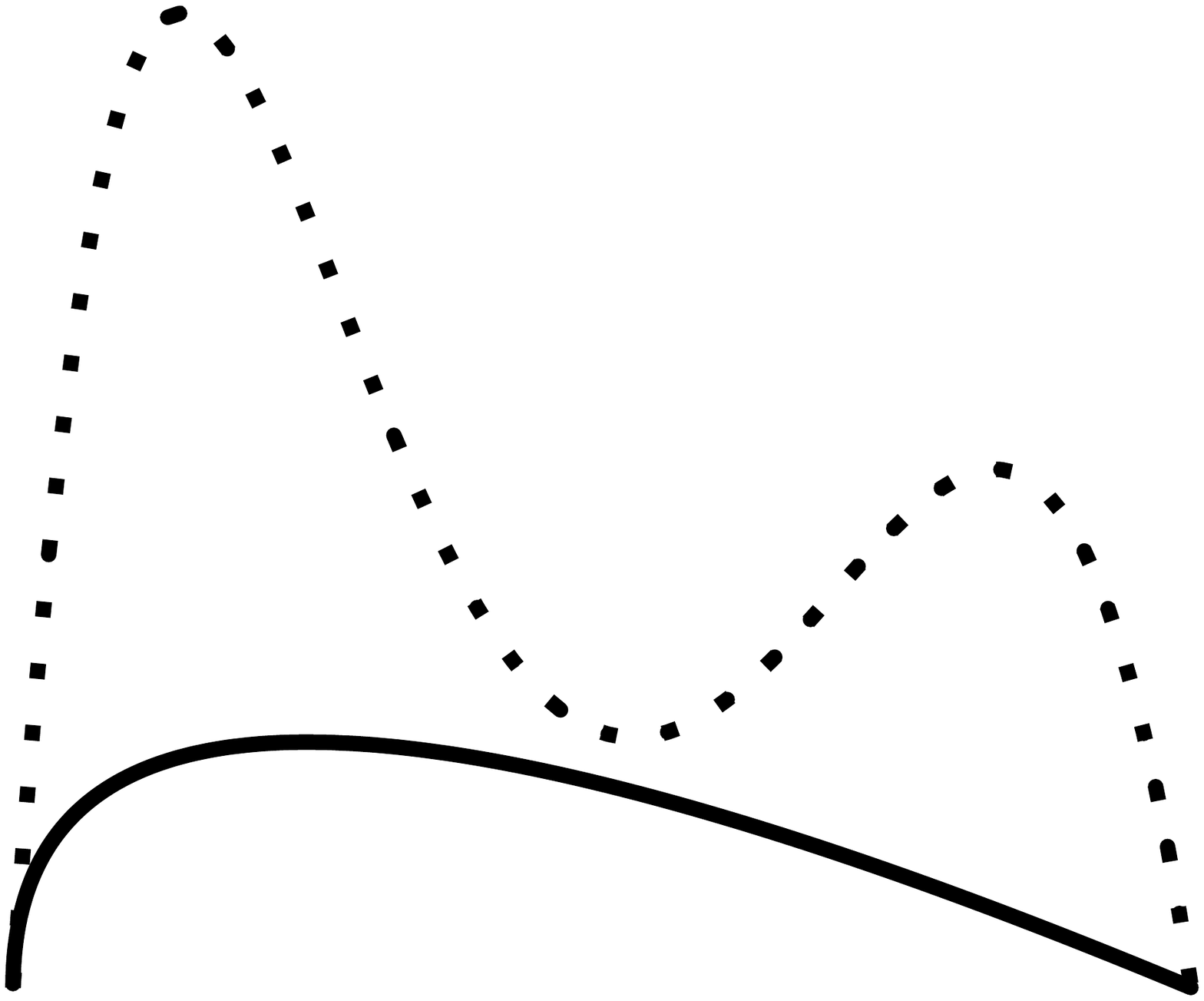}}
\hspace{0.2in}
			\subfloat[]{\includegraphics[width=0.25\textwidth]{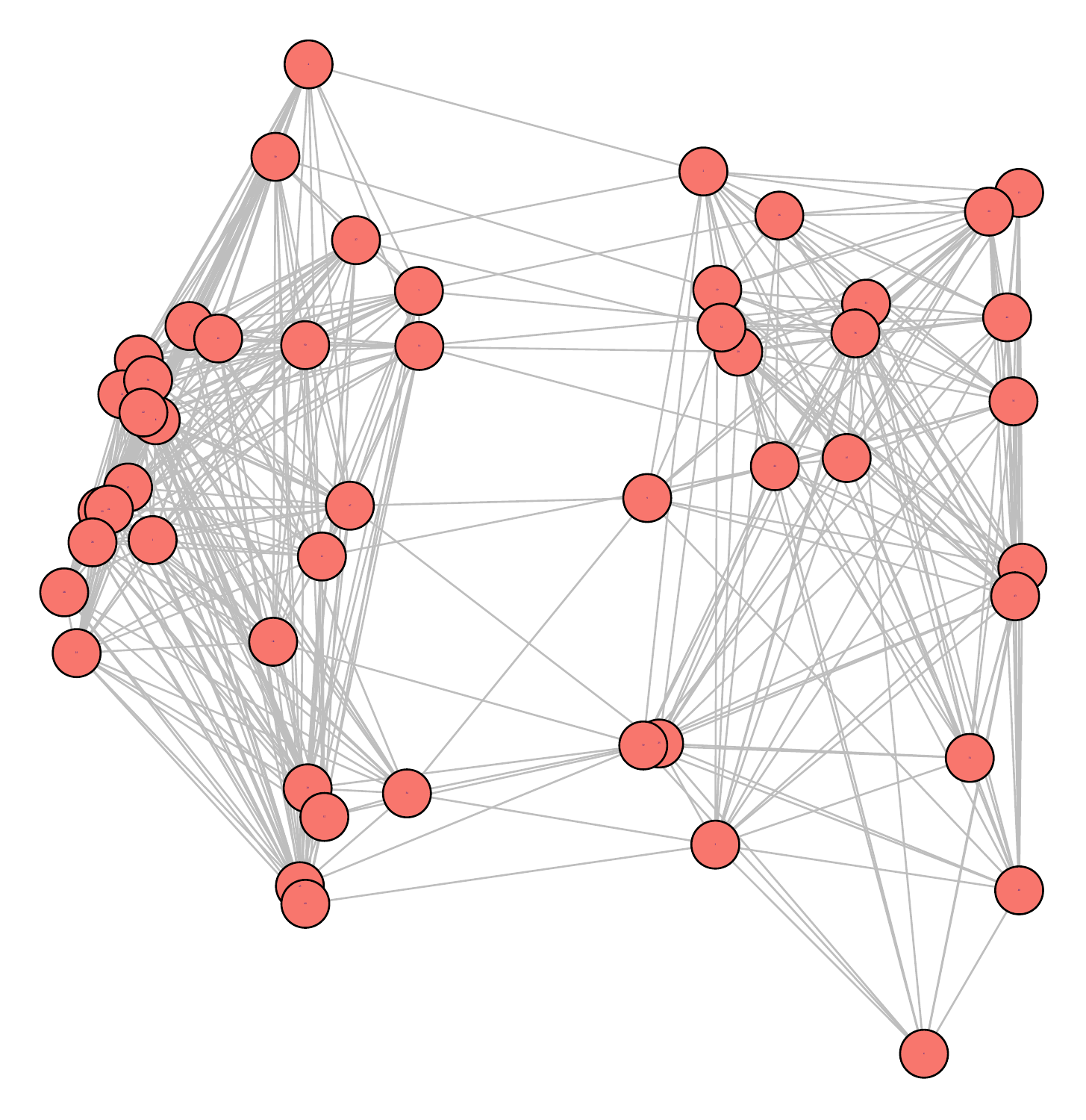}}
			 \caption{(a): {\em Unobserved} density of a Beta mixture distribution on $[0,1]$. (b): {\em Unobserved} latent position distribution on Hardy-Weinberg curve (c): {\em Observed} realization of a random dot product graph generated from latent positions on the Hardy-Weinberg curve.}
	\label{fig:HW1mix}
\end{figure}

\begin{remark}[Stochastic Block Models as LSMs]\label{rem:SBMs_as_LSMs} We emphasize that latent-structure models with one-dimensional structural support can encompass stochastic block models with fixed block probability vectors, because such stochastic block models have a latent position distribution $F$ that is a discrete mixture. More precisely, let the support of $F$ be given by $k$ distinct points $\{x_1, \dots, x_k\}$ in $\mathbb{R}^d$ with weights $\{a_1, \dots, a_k\}$ respectively. Suppose these $k$ distinct points lie on a smooth, non-self intersecting curve $\bC$, with $p:[0,1] \rightarrow \bC$ its arclength parameterization. Let $G$ be a distribution on $[0,1]$ supported on the set of points $p^{-1}(x_i)$ with weights $\{a_1, \dots, a_k\}$. Then $F=G(p^{-1})$. Note that $\bC$ need {\em not} be unique. If the estimation task is that of determining these point masses and their weights, the particular choice of $\bC$ is immaterial. For more on efficient estimation of weights in a stochastic block model, we refer the reader to \cite{tang_sbm_eff}.
\end{remark}

 	\section{Inference on Latent Structure Models: Summary of spectral methods for RDPGs}\label{sec:prior_results_RDPG}
 	Since a latent structure random graph is necessarily a random dot product graph, our program for inference on a latent structure models is to follow an algorithm that leverages the accuracy of spectral embeddings for latent position estimation in  random dot product graphs. First, we embed the adjacency matrix of the LSM into the correct embedding dimension $d$ (the rank of the RDPG) or we embed it into a suitable estimate $\hat{d}$ of this dimension.  This yields a collection of estimates $\hat{X}_i$ of the true latent positions $X_i$.  Second, we consider rotating (due to nonidentifiability) and projecting (due to the noise inherent in these estimates) the $\hat{X}_i$ estimates on to the curve $\bC$. Denote these rotated and projected estimated latent positions by $\breve{X}_i$; by construction, they lie on $\bC$.
 	Third, since the true latent positions $X_i$ are independent and identically distributed with distribution $F_{\theta}$, under suitable regularity conditions, classical maximum likelihood estimation using the $X_i$ points yields efficient estimation of $\theta$ with a variance of order $1/n$. But because the rotated and projected estimated latent positions $\breve{X}_i$ are sufficiently ``close" to the true latent positions, we next treat these estimates as the actual ``data"---that is, we regard the $\breve{X}_i$ as appropriate substitutes for the actual $X_i$ points, even though the latter are i.i.d from the distribution $F$, and the former are decidedly not. Finally, we conduct $M$-estimation of parameters of $F$ using the $\breve{X}_i$, and when considering a one-dimensional latent structure model in which $G=p^{-1}(F)$ belongs to some parametric or nonparametric family in $[0,1]$, we apply classical parametric or nonparametric estimation techniques to
 	$\hat{Y}_i=p^{-1}(\breve{X}_i)$.
 	Note that the $\hat{Y}_i$ points are not independent. They are not un-noisy. They are pullbacks of rotations of projections. Despite these limitations, 
 	under reasonable regularity conditions on a parametric one-dimensional latent structure model, $M$-estimation for the parameters $\theta$ of $G$ using the points $Y_i$ has the same parametric rate of convergence to the true value $\theta_0$ as we might obtain with the pullbacks of the true latent positions $X_i$.

	As we noted earlier, this parametric rate of convergence is somewhat surprising,
and is part of a larger wish list for spectral estimates. Our prior work demonstrates that the spectrally-estimated latent positions are consistent and asymptotically normal, and we prove here that these latent position estimates can generate parametric $M$-estimators that are asymptotically {\em efficient}.

 	We begin by defining the adjacency spectral embedding of a random dot product graph.
 	\begin{definition} [Adjacency spectral embedding (ASE)]\label{def:ASE}
 		Given a positive integer $d \geq 1$, the {\em adjacency spectral embedding} (ASE) of $\mathbf{A}$ into
 		$\mathbb{R}^{d}$ is given by $\hat{{\bf X}}={\bf U}_{\mathbf{A}}
 		{\bf S}_{\mathbf{A}}^{1/2}$ where
 		$$|{\bf A}|=[{\bf U}_{\mathbf{A}}|{\bf U}^{\perp}_{\mathbf{A}}][{\bf
 			S}_{\mathbf{A}} \bigoplus {\bf S}^{\perp}_{\mathbf{A}}][{\bf
 			U}_{\mathbf{A}}|{\bf U}^{\perp}_{\mathbf{A}}]^{\top}$$ is the spectral
 		decomposition of $|\bf{A}| = (\bf{A}^{\top} \bf{A})^{1/2}$ and
 		$\mathbf{S}_{\mathbf{A}}$ is the diagonal matrix of the $d$ largest eigenvalues
 		of $|\mathbf{A}|$ and $\mathbf{U}_{\mathbf{A}}$ is the $n \times d$ matrix whose
 		columns are the corresponding eigenvectors. 
 	\end{definition}
 	
 	We now describe a consistency result in the $2 \to \infty$ norm that provides uniform control of deviations between the estimated and true latent positions \cite{lyzinski15_HSBM}. This uniform control can matter significantly for the subsequent inference task, as we describe below.
 	Furthermore, an analogous result of this type can and has been extended to much more general random matrix perturbations, including covariance matrix estimation \cite{cape2toinfty}. We state our $2 \to \infty$ bound here in a form tailored to the RDPG setting.
 	
 	\begin{theorem}[Theorem 5, \cite{lyzinski15_HSBM}]\label{thm:minh_sparsity}
 		Let $\bA_n \sim \mathrm{RDPG}(\bX_n)$ for $n \geq 1$ be a sequence of random dot product graphs where the $\bX_n$ is assumed to be of rank $d$ for all $n$ sufficiently large. Let $\mathbf{P}_n = \mathbf{X}_n \mathbf{X}_n^{\top}$ and let $\delta_n = \max_{i} \sum_{j} \mathbf{P}_{n,ij}$ be the maximum expected degree. Denote by $\hat{\bX}_n$ the adjacency spectral embedding of $\bA_n$ and let $(\hat{\bX}_{n})_{i}$ and $(\bX_n)_{i}$ be the $i$-th row of $\hat{\bX}_n$ and $\bX_n$, respectively. Let $E_n$ be the event that there
 		exists an orthogonal transformation $\bW_n \in \mathbb{R}^{d \times d}$ such that
 		\begin{equation*}
 		\max_{i} \| (\hat{\bX}_n)_{i} - \bW_n (\bX_n)_{i} \| \leq 
 		\frac{C d^{1/2} \log^2{n}}{\delta_n^{1/2}}
 		\end{equation*}
 		where $C > 0$ is some fixed constant. Then $E_n$ occurs asymptotically almost surely; that is, $\Pr(E_n) \rightarrow 1$ as $n \rightarrow \infty$.
 	\end{theorem}
 	
 	
 	Having established that the estimated latent positions are consistent in this $2 \to \infty$ norm, we next point out that the latent position estimates are asymptotically normal. Specifically, for a $d$-dimensional random dot product graph with i.i.d latent positions, there exists a sequence of $d \times d$ orthogonal matrices $\bW_n$ such that for any row index $i$, $\sqrt{n}(\bW_n (\Xhat_n)_{i} - (\bX_n)_i)$ converges as $n \rightarrow \infty$ to a mixture of multivariate normals (see \cite{athreya2013limit}).
 	\begin{theorem}[Central Limit Theorem for rows of ASE; Theorem 1, \cite{athreya2013limit}]\label{thm:clt_orig_but_better}
 		Let $(\bA_n, \bX_n) \sim \mathrm{RDPG}(F)$ be a sequence of adjacency matrices and associated latent positions of a $d$-dimensional random dot product graph according to an inner product distribution $F$. Let $\Phi(\bx,\bSigma)$ denote the cdf of a (multivariate)
 		Gaussian with mean zero and covariance matrix $\bSigma$,
 		evaluated at $\bx \in \R^d$. Then
 		there exists a sequence of orthogonal $d$-by-$d$ matrices
 		$( \Wn )_{n=1}^\infty$ such that for all $\bm{z} \in \R^d$ and for any fixed index $i$,
 		$$ \lim_{n \rightarrow \infty}
 		P\left[ n^{1/2} \left( \Xhat_n \Wn - \bX_n \right)_i
 		\le \bm{z} \right] 		= \int_{\supp F} \Phi\left(\bm{z}, \bSigma(\bx) \right) dF(\bx), $$
 		where
 		\begin{equation}
 		\label{def:sigma}\bSigma(\bx) 
 		= \Delta^{-1} \E\left[ (\bx^{\top} X_1 - ( \bx^{\top} X_1)^2 ) X_1 X_1^{\top} \right] \Delta^{-1}; \quad \text{and} \,\, \Delta = \mathbb{E}[X_1 X_1^{\top}] \,\, \text{with $X_1 \sim F$}. 
 		\end{equation}
 	\end{theorem}
 	We recall that when $F$ is a mixture of $K$ point masses, i.e., $F = \sum_{k=1}^{K} \pi_{k} \delta_{\nu_k}, \pi_1, \dots, \pi_K > 0, \sum_{k} \pi_k = 1$, then $(\mathbf{X}, \mathbf{A}) \sim \mathrm{RDPG}(F)$ is a $K$-block stochastic blockmodel graph. Thus, for any fixed index $i$, the event that $\bX_i$ is assigned to block $k \in \{1,2,\dots,K\}$ has non-zero probability and hence one can condition on the block assignment of $\bX_i$ to show that the conditional distribution of $\sqrt{n}(\mathbf{W}_n (\hat{\bX}_n)_{i} - (\bX_n)_i)$ converges to a multivariate normal. More specifically,
 		\begin{equation}
 		P\Bigl\{\sqrt{n}(\mathbf{W}_n \hat{\bX}_n - \mathbf{X}_n)_{i} \leq \bm{z}  \mid \bX_i = \nu_k \Bigr\} \longrightarrow \Phi(\bm{z}, \Sigma_k)
 		\end{equation}
 		where $\Sigma_k = \Sigma(\nu_k)$ is as defined in Eq.~\eqref{def:sigma}. 
 	
Strictly speaking, to prove the efficiency of $M$-estimates of underlying parameters in a one-dimensional latent structure model, the asymptotic normality of the estimated latent positions is not required, but the  classical nature of such a central limit theorem warrants its inclusion here. What {\em is} required to prove such efficiency, though, is the following key empirical process result from \cite{tang14:_nonpar}, below. This empirical process result implies the uniform convergence of scaled sums of differences of functions of estimated and true latent positions, provided the functions belong to a sufficiently regular class.
 	We first recall certain definitions, which we
 	reproduce from \cite{vaart96:_weak}. Let $X_i, 1 \leq i \leq n$ be
 	identically distributed random variables on a measure space
 	$(\mathcal{X}, \mathcal{B})$, and let $P_n$ be their
 	associated {\em empirical measure}; that is, $P_n$ is the
 	discrete random measure defined, for any $E \in \mathcal{B}$, by
 	$$P_n(E)=\frac{1}{n} \sum_{i=1}^n 1_{E}(X_i).$$
 	Let $P$ denote the common distribution of the random variables $X_i$,
 	and suppose that $\mathcal{F}$ is a class of measurable, real-valued
 	functions on $\mathcal{X}$.  The {\em $\mathcal{F}$-indexed empirical
 		process} $\mathbb{G}_n$ is the stochastic process
 	\begin{equation*}
 	f \mapsto \mathbb{G}_n(f)=\sqrt{n}(P_n -P)f=
 	\frac{1}{\sqrt{n}} \sum_{i=1}^n \Bigl(f(X_i)-
 	\mathbb{E}[f(X_i)]\Bigr). 
 	\end{equation*}
 	Under certain conditions, the empirical process $\{\mathbb{G}_n(f): f
 	\in \mathcal{F}\}$ can be viewed as a map into
 	$\ell^{\infty}(\mathcal{F})$, the collection of all uniformly bounded
 	real-valued functionals on $\mathcal{F}$.  In particular, let
 	$\mathcal{F}$ be a class of functions for which the empirical process
 	$\mathbb{G}_n=\sqrt{n}(P_n-P)$ converges to a limiting
 	process $\mathbb{G}$ where $\mathbb{G}$ is a tight Borel-measurable
 	element of $\ell^{\infty}(\mathcal{F})$ (more specifically a Brownian
 	bridge). Then $\mathcal{F}$ is said to be a {\em $P$-Donsker class}.
 	\begin{theorem}[Theorem 4, \cite{tang14:_nonpar}]
 		\label{thm:u-statistics}
 		Let $({\bf X}_n, {\bf A}_n)$ for $n = 1,2,\dots,$ be a sequence of $d$-dimensional
 		$\mathrm{RDPG}(F)$. Let $\mathcal{F}$ be a collection of
 		(at least) twice continuously differentiable functions on $\mathrm{supp} \,F $ with
 		\begin{equation*}
 		\sup_{f \in \mathcal{F}, X \in \mathrm{supp}\,F} \|(\partial f)(X)\| <
 		\infty; \qquad \sup_{f \in \mathcal{F}, X \in \mathrm{supp}\,F}
 		\| (\partial^{2} f)(X) \| < \infty.
 		\end{equation*} 
 		Furthermore, suppose 
 		$\mathcal{F}$ is such that
 		$\mathbb{G}_n=\sqrt{n}(P_n-P)$ converges to $\mathbb{G}$, a
 		$P$-Brownian bridge on $\ell^{\infty}(\mathcal{F})$. Then there exists a sequence of orthogonal matrices $\mathbf{W}_n$ such that as $n \rightarrow \infty$, 
 		\begin{equation}
 		\label{eq:fcorclt1}
 		\sup_{f \in \mathcal{F}} \,\, \Bigl|\frac{1}{\sqrt{n}} \sum_{i=1}^n
 	\Bigl(f(\mathbf{W}_n \hat{X}_i) - f(X_i)\Bigr)\Bigr| \rightarrow 0,
 	\end{equation}
 	where $\{\hat{X}_i\}_{i=1}^{n}$ are the rows of $\hat{\mathbf{X}}_n$. 
 	Therefore, the $\mathcal{F}$-indexed empirical process
 		\begin{equation}
 		f \in \mathcal{F} \mapsto \hat{\mathbb{G}}_{n} f =
 		\frac{1}{\sqrt{n}}\sum_{i=1}^{n} \Bigl(f( \mathbf{W}_n \hat{X}_i) - \mathbb{E}[f(X_i)]\Bigr)
 		\end{equation}
 		also converges to $\mathbb{G}$ on $\ell^{\infty}(\mathcal{F})$. 
 	\end{theorem}
 	Theorem~\ref{thm:u-statistics} is in essence a functional central
 	limit theorem for the estimated latent positions $\{\hat{X}_i\}$ in
 	the RDPG setting, and we emphasize that for any $n$, the
 	$\{\hat{X}_i\}_{i=1}^{n}$ are not jointly independent random variables, and therefore
 	Theorem~\ref{thm:u-statistics} is a functional central limit
 	theorem for {\em dependent} data. Due to the non-identifiability of random dot
 	product graphs, there is an explicit dependency on a sequence of
 	orthogonal matrices $\mathbf{W}_n$. The main technical result in Theorem~\ref{thm:u-statistics} is Eq.~\eqref{eq:fcorclt1}, which we use to show the asymptotic normality of $M$-estimation for the parameters of LSMs in Section~\ref{sec:asymp_eff_LSM}.
 	 	\section{Asymptotically efficient $M$-estimation in latent structure models}\label{sec:asymp_eff_LSM}
 	 	Suppose our parameter space is $\bm{\Theta} \subset \mathbb{R}^{l}$ and assume this is a open, nonempty, connected set with compact closure. We denote a particular parameter value by $\theta=(\theta_1, \dots, \theta_j)$. Let $F$ be an inner product distribution with $\supp F \subset B(0, R_0)$ where $B(0,R_0)$ is the ball of radius $R_0>0$ about $0$ in $\bR^d$.  We assume that $F$ represents the cumulative distribution function of a one-dimensional latent structure model with known support $\bC$, with $\bC$ an LSM-regular curve of minimal subspace dimension $d$. Let $p:[0,1]\rightarrow \bC$ be the smooth and smoothly invertible arclength parameterization of $\bC$.  For the underlying distribution of our LSM, let $G_{\theta}: \theta \in \bm{\Theta}$ be a parametric family of cumulative distribution functions supported on the unit interval $[0,1]$, and with density $g(\cdot,\theta)$. Let $\pi:\mathbb{R}^d \rightarrow \bC$ be the distance-minimizing projection of a point in $\mathbb{R}^d$ to $\bC$. 
 	 	
 	 	Since $\bC$ is an LSM-regular curve, there exists a tubular neighborhood $T_{\bC}(R)$ of radius $R>0$ about $\bC$ for which the projection $\pi$ onto $\bC$ is well-defined and sufficiently smooth. Therefore, there exist $R>R_2>R_1>0$ for which we can construct a sufficiently smooth function $f(x, \theta): \mathbb{R}^{d \times l} \rightarrow \mathbb{R}$ satisfying   	 \begin{equation}\label{eq:mollified_log_likelihood}
        f(x, \theta)=\begin{cases}
        \log g(p^{-1}(\pi(x)), \theta) & \textrm{ if }
 	 	x \in T_C(R_1)\\
        0  & \textrm{ if } x \notin T_{\bC}(R_2)
        \end{cases}
        \end{equation}
        Observe that such a function can always be constructed using mollifiers; that is, we can write
        $$f(x,\theta)=\log g(p^{-1}(\pi(X), \theta) \cdot h(x)$$
        where $h(x)$ is a mollifier---a smooth function that is identically equal to 1 on $T_{\mathcal{C}}(R_1)$ and that vanishes outside of $T_{\mathcal{C}}(R_2)$. Now, the first radius, $R_0$, is that of a ball sufficient to encompass the necessarily compact support of $F$. Next, $R>0$ is chosen sufficiently small so that the projection operator $\pi$ onto the closest point in $\mathcal{C}$ is well-defined in the tubular neighborhood $T_{\mathcal{C}}(R)$. We stress that $R$ depends only on $\mathcal{C}$. Finally, $R_1$ and $R_2$, with $0<R_1<R_2<R$, are defined so that a mollification of $\log(p^{-1}(\pi(x))$ can be constructed with the following properties: within the tubular neighborhood $T_{\mathcal{C}}(R_1)$, the mollification $f$ is equal to $\log(g(p^{-1}(\pi(x)), \theta)$. Outside of $T_{\mathcal{C}}(R_2)$, $f$ vanishes. Observe that $f$ is necessarily compactly supported.
 	 	
 	 	Let $f_j(x, \theta): \mathbb{R}^{d \times l} \rightarrow \mathbb{R}$ be defined as follows:
 	 	$$f_j(x, \theta)=\frac{\partial f}{\partial \theta_j} (x, \theta)$$
 	 	Because of how $f$, above, is defined, it is immediate that
 	 	$$f_j(x,\theta)=\frac{\partial \log g(p^{-1}(\pi(x)),\theta)}{\partial \theta_j}$$
 	 	for $x \in T_C(R_1)$, and, as before, $f_j(x,\theta)=0$ for all $\theta \in \bm{\Theta}$ and all $x$ outside of $T_C(R_2)$.
 	 	We require that $f_j$ be twice continuously differentiable with respect to $\theta_j$ for $j \in 1, \dots, l$ and $x_1, \dots x_d$.
 	 	
 	 	Since we are considering maximum likelihood estimates for $\theta$, which are equivalently expressible as minimum contrast estimates (where we minimize the sum of the negations of the log likelihoods), we assume that $\hat{\theta}_n$ is given by 
 	 	$$\hat{\theta}_n=\arg \min \left[ -\frac1n\sum_{i=1}^n f(X_i, \theta)\right]$$
 	 	Suppressing, for notational convenience, the dependence of $\hat{\theta}_n$ on $n$, we assume that $\hat{\theta}$ satisfies
 	 	$$\frac{1}{n} \sum_{i=1}^n \Psi(X_i, \hat{\theta})=0$$
 	 	where 
 	 	$$\Psi(X, \theta)=\left(\frac{\partial \log g(p^{-1}(\pi(X)),\theta)}{\partial \theta_1}, \cdots, \frac{\partial \log g(p^{-1}(\pi(X)),\theta)}{\partial \theta_l}\right)=(f_1(X, \theta), \cdots, f_l(X, \theta))$$
 	 	and $X$ is a random draw from the latent structure distribution $F_{\theta}$.
 	 	We assume the following standard regularity conditions on $f(x, \theta)$ and $f_j(x, \theta)$ (see \cite[p.~328,  p.~384]{Bickel_Doksum}). We reproduce these familiar conditions here to reinforce the fact that, for our main theorem establishing asymptotic efficiency for $M$-estimates of graph parameters using the estimated latent positions $\hat{X}_i$ in place of the true latent positions $X_i$, the standard regularity conditions still suffice. 
 	 	\begin{enumerate}[(a)]
 	 		\item (Uniqueness) The equation
 	 		\begin{equation}
 	 		\label{eq:regularity_unique_A1}
 	 		\int \Psi(x, \theta) dF_{\theta_0}(x)=0
 	 		\end{equation}
 	 		for $\theta \in \Theta$, has a unique solution at $\theta = \theta_0$.
 	 		\item ($L^2$ boundedness on partial derivatives, nonsingularity of the Hessian, and uniform convergence of sample means) If $X \sim F_{\theta_0}=G_{\theta_0}(p^{-1})$, then \begin{equation}\label{eq:L2norm_req}
 	 		\mathbb{E}_{\theta_0}(|\Psi(X, \theta_0)|^2) < \infty;
 	 		\end{equation}
 	 		and for all $\theta$ and $X \sim F_{\theta}$,
the $l \times l$ matrix of second partial derivatives of $f$ denoted by
 	 		 $$D\Psi(X, \theta); \quad (D \Psi(X, \theta))_{jk} = \frac{\partial f_j(X, \theta_0)}{\partial \theta_k}$$
 	 	     satisfies $\|\mathbb{E}_{\theta_0}[D \Psi(X, \theta_0)]\| < \infty$ and that $\mathbb{E}_{\theta_0}[D \Psi(X, \theta_0)]$ is invertible. (In our specific case, where $f$ is the log-likelihood, the negation of this matrix is the familiar Fisher information.)
 	 		
 	 	 Next, if $\epsilon_n$ is a positive sequence of real numbers converging to zero, then
 	 		\begin{equation}
 	 		\label{eq:fisher-info-convergence}
 	 		P_{\theta_0}\left(\sup_{t}\left\{\big|\frac1n \sum_{i=1}^n \left[D\Psi (X_i, t)- D\Psi (X_i, \theta_0)\right]\big|: |t-\theta_0|< \epsilon_n\right\}\right) \rightarrow 0
 	 		\end{equation}
 	 		as $n \rightarrow \infty$
 	 		\item (Sufficient conditions for consistency of a minimum contrast estimate) The function $Q(\theta_0, \theta)$ defined by
 	 		$$Q(\theta_0, \theta)=\mathbb{E}_{\theta_0}\left[-f(X, \theta)\right]$$
 	 		has a unique minimum at $\theta_0$, and 
 	 		\begin{equation}\label{eq:unique_min}
 	 		\inf\{Q(\theta_0, \theta): |\theta-\theta_0|\geq\epsilon\}>Q(\theta_0, \theta_0) \, \forall \epsilon>0
 	 		\end{equation}
 	 		where $\| \cdot \|$ denotes the Euclidean norm in $\mathbb{R}^l$.
 	 		Furthermore, we have the following uniform weak law:
 	 		\begin{equation}\label{eq:consistency_req}
 	 		P_{\theta_0}\left(\sup\left\{\big|\frac1n \sum_{i=1}^n (-f(X_i, \theta))-Q(\theta_0, \theta)\big|: \theta \in \bm{\Theta}\right\}\right) \rightarrow 0
 	 		\end{equation}
 	 		where $P_{\theta_0}$ connotes the probability computed when $\theta=\theta_0$.
 	 	\end{enumerate}
 	 	Let $X_1, \dots, X_n$ be i.i.d $F_{\theta_0}=G_{\theta_0}(p^{-1})$ on $\mathcal{C}$ be our collection of latent positions, organized by rows into the latent position matrix $\bX$.
 	 	Since $X_i$ are i.i.d $F_{\theta_0}=G_{\theta_0}(p^{-1})$, we observe that the maximum likelihood estimate for $\theta_0$, denoted $\hat{\theta}_n$, is, under the above regularity assumptions, well-defined, consistent, and asymptotically normal, with a variance given by the inverse of the Fisher information (again, see \cite{Bickel_Doksum} for a proof of this quintessentially classical result). Namely, suppose we define (suppressing for notational convenience a dependence here on sample size $n$) $\hat{\theta}$ via
 	 	\begin{equation}
 	 	\label{eq:true_MLE}
 	 	\hat{\theta}=
 	 	\arg \, \min \left[-\frac1n \sum_{i=1}^n f(X_i, \theta)\right]
 	 	\end{equation}
 	 	then $\hat{\theta}$ is consistent for $\theta_0$, and furthermore
 	 	\begin{equation}\label{asymp_normal_MLE}
 	 	\sqrt{n}(\hat{\theta}-\theta_0) \rightarrow \mathcal{N}(0, I^{-1}(\theta_0))
 	 	\end{equation}
 	 	where $I(\theta_0)=-E_{\theta_0}[D\Psi(X, \theta_0)]$ is the Fisher information matrix.
 	 	 
 	 	Next, suppose $\bA$ is the adjacency matrix of a random dot product graph with this latent position matrix $\bX$, and let $\hat{\bX}$ be the adjacency spectral embedding of $\bA$.
 	 	Let $\bW_n$ be the orthogonal transformation satisfying
 	 	\begin{equation}\label{eq:Proc_transform}
 	 	\min_{W \in \mathcal{O}(d\times d)}\|\hat{\bX}\bW-\bX\|.
 	 	\end{equation} 
 	 	Let $\hat{\bX}_{r}=\hat{\bX}\bW_n$ denote the properly rotated latent positions.  For convenience, we will employ a slight abuse of notation and use $\hat{\bX}$ to denote this rotated version of our latent positions, so that in what follows below, 
 	 	$\hat{\bX}=\hat{\bX}_{r}$.
 	 	Let $\{\hat{X}_i\}_{i=1}^{n}$ be the rows of $\hat{\bX}$, and suppose that $\tilde{\theta}$ is defined analogously to the maximum likelihood estimate, except with the $\hat{X_i}$ points in place of the true latent positions:
 	 	\begin{equation}
 	 	\label{eq:theta_tilde}
 	\tilde{\theta}=
 	\arg \min \left[-\frac1n \sum_{i=1}^nf(\hat{X}_i, \theta)\right]
 	\end{equation}
 	We emphasize that $\tilde{\theta}$ is an $M$-estimate for $\theta$ determined not by the unobserved true latent positions $X_i$, but rather their estimates $\hat{X}_i$.
 	Note that $\tilde{\theta}$ satisfies
 	$$\frac{1}{n}\sum_{i=1}^n f_j(\hat{X}_i, \tilde{\theta})=0$$
 	 	for all $j$.
 	 	
 	 	Our principal result is that a minimum contrast estimate involving the estimated latent positions possesses the same desirable asymptotic properties as the classical maximum likelihood estimator $\hat{\theta}$ that is a function of the true i.i.d latent positions. In particular, we will show that $$\sqrt{n}(\tilde{\theta}-\theta_0) \rightarrow \mathcal{N}(0, I^{-1}(\theta_0)),$$ which is the content of Theorem \ref{thm:parametric_rate_xhats} below. The proof of this result depends on two pieces: first, a consistency result, which is that $\hat{\theta}-\tilde{\theta}$ converges to zero in probability; and second, an asymptotic normality result under a $\sqrt{n}$ scaling. Under sufficient smoothness conditions for the log likelihoods, consistency of the $M$-estimates follows from the consistency of the adjacency spectral embedding for the true latent positions---that is, from the $2 \to \infty$-norm result of Theorem~\ref{thm:minh_sparsity}. 
 	 The asymptotic normality result, on the other hand, requires a stronger convergence, precisely because we need to show that
 	 $\sqrt{n}(\tilde{\theta}-\theta)$ has a limiting normal distribution.  Thus the asymptotic normality is consequence of our Donsker analogue, Eq.~\eqref{eq:fcorclt1}, which gives a uniform convergence to zero of the scaled sum
 	 $\left[\frac{1}{\sqrt{n}}\sum_{i=1}^n f(\hat{X}_i)-\frac{1}{\sqrt{n}}\sum_{i=1}^n f(X_i)\right]$. To guarantee a parametric rate for our $M$-estimates, it is crucial that this convergence to zero occur even when the scaling is of order $1/\sqrt{n}$, not $1/n$. 
 	 
 	 We begin with the more straightforward consistency result.
 	 	\begin{lemma}\label{lem:hat_theta_to_tilde_theta}
 	 		Let $\hat{\theta}$ and $\tilde{\theta}$ be as defined in Eqs. \eqref{eq:true_MLE} and \eqref{eq:theta_tilde}, above. Let $c>0$ be any positive constant.  Then
 	 		$$P_{\theta_0}(|\hat{\theta}-\tilde{\theta}|>c) \rightarrow 0$$
 	 	\end{lemma}
\begin{proof}
	We first impose a certain uniform continuity requirement on the parametric family of distributions $G_{\theta}$. That is, 
	suppose that $\ell(x,y):\mathbb{R}^2 \rightarrow \mathbb{R}$ is such that any $\epsilon>0$, there exists $\delta>0$ for which $\|x-y\|<\delta$ guarantees $|\ell(x,y)|<\epsilon$. 
	Letting $\lambda$ be Lebesgue measure in $\mathbb{R}^1$, we require that for $\lambda$-almost all $x, y \in [0,1]$, all $\theta \in \bm \Theta$, and all $j \in \{1, \dots, l\}$, 
	\begin{equation}\label{eq:Lipschitz_like}
	|\log g(x, \theta)-\log g(y, \theta)| \leq \ell(x,y)
	\end{equation}
	The smoothness of the map $p^{-1}$ onto the curve $\mathcal{C}$ defining the structural support for the latent structure model then ensures that this same uniform continuity property holds for $f(x, \theta)$.  By Theorem~\ref{thm:minh_sparsity}, we note that with probability tending to one as $n \rightarrow \infty$,
	\begin{equation}
	\label{eq:2toinf_consistency_MLE}
	\left|\frac{1}{n}\sum_{i=1}^nf(X_i, \theta)-\frac{1}{n} \sum_{i=1}^n f(\hat{X}_i, \theta)\right|<\ell(\hat{X}_i,X_i) \rightarrow 0
	\end{equation}
	because of the $2 \to \infty$ bound given in Thm. \ref{thm:minh_sparsity}.
	Thus, the sequence of functions
	$$\frac{1}{n} \sum_{i=1}^n f(\hat{X}_i, \theta)-\frac{1}{n} \sum_{i=1}^n f(X_i, \theta)$$
	converges in probability to $0$ {\em uniformly in $\theta$}.
	Because of this, Eq.~\eqref{eq:consistency_req} guarantees that 
	$$\frac{1}{n} \sum_{i=1}^n f(\hat{X}_i, \theta)$$ 
	converges uniformly in probability to $Q(\theta, \theta_0)$ as well.  With an argument exactly analogous to that in \cite[\S~5.2]{Bickel_Doksum}, this implies that $\tilde{\theta}$ converges to $\theta_0$ in probability as well.
\end{proof}
 	 	We remark that many distributions, including the $\textrm{Beta}(a,b)$ family, vanish at the endpoints of $[0,1]$, and hence a truncated version of the log-likelihood for these distributions will satisfy the uniform continuity requirement as long as we restrict ourselves to compact parameter spaces. 
 	 	
 	 		We will use this to show the stronger result that
 	 		$$\sqrt{n}(\hat{\theta}-\tilde{\theta})$$ converges to zero in probability. Once we have proved this stronger result, we can write
 	 		$$\sqrt{n}(\tilde{\theta}-\theta_0)=\sqrt{n}(\tilde{\theta}-\hat{\theta}) + \sqrt{n} (\hat{\theta}-\theta_0)$$
 	 		As we discussed earlier, classical results on maximum likelihood estimation ensure that the latter of these two summands converges to a normal distribution, and we will show that the first summand converges in probability to zero.  Slutsky's Theorem then establishes the asymptotic efficiency of the $M$-estimate obtained with the estimated latent positions $\hat{X}_i$, which is stated next.
 	 	\begin{theorem}\label{thm:parametric_rate_xhats}
 	 		Suppose $X_i \sim$ i.i.d $F_{\theta_0}$ are latent positions of a latent structure model satisfying the regularity assumptions delineated above.  Let $\bA$ be the adjacency matrix of the random dot product with latent positions $\bX$, and let $\hat{\bX}$ be the suitably-rotated adjacency spectral embedding of $\bA$.
 	 		Let $\hat{\theta}$ and $\tilde{\theta}$ satisfy
 	 		$$\hat{\theta}=\arg\min \left[-\frac1n\sum_{i=1}^n f(X_i, \theta))\right], \qquad \, \tilde{\theta}=\arg\min \left[-\frac1n\sum_{i=1}^n f(\hat{X}_i, \theta)\right].$$ Then
 	 	$$\sqrt{n}(\tilde{\theta}-\theta_0) \rightarrow \mathcal{N}(0, I^{-1}(\theta_0)), \textrm{ where } I(\theta_0)_{jk}=-\mathbb{E}_{\theta_0} \left(\frac{\partial^2 f(X,\theta)}{\partial \theta_j \partial \theta_k} \Big \rvert_{\theta = \theta_0} \right)$$
 	 	denotes the Fisher information matrix.
 	 		\end{theorem}
 	 	\begin{proof}
 	 		Observe that $\sqrt{n}(\hat{\theta}-\theta_0)$ converges to a normal distribution with mean zero and variance $I^{-1}(\theta_0)$ (see, for example, \cite[6.2.2]{Bickel_Doksum}).  Thus, it remains to show that
 	 		$$\sqrt{n}(\hat{\theta}-\tilde{\theta}) \rightarrow 0$$
 	 		in probability.  To this end, first note that for every $j \in 1 \dots l$,  $f_j(\cdot, \theta)$ is a compactly supported, twice-continuously differentiable function on $\mathbb{R}^d$. Letting 
 	 		$\mathcal{F}=\{f_j(\cdot, \theta): \theta \in \bm{\Theta}\},$
 	 	 we find that this collection of functions is a Donsker class \cite{vaart96:_weak}. As such, Eq.~\eqref{eq:fcorclt1} guarantees that 
 	 	$$\sup_{\theta \in \bm{\Theta}} \Big \lvert \frac{1}{\sqrt{n}}\sum_{i=1}^n\bigl(f_j(\hat{X}_i, \theta)-f_j(X_i, \theta)\bigr)\Big \rvert \rightarrow 0$$
 	 	(Note that because we assume that the matrix of estimated latent position has been appropriately rotated, we can suppress here the sequence of orthogonal transformations that are part of Theorem~\ref{thm:u-statistics}.)
 	 	Therefore, we have that
 	 	$$ \Big \lvert \frac{1}{\sqrt{n}}\sum_{i=1}^n \bigl(f_j(\hat{X}_i, \tilde{\theta})-f_j(X_i, \tilde{\theta} \bigr) \Big \rvert \rightarrow 0$$
 	 	Note that  
 	 	$\frac{1}{\sqrt{n}}\sum_{i=1}^n f_j(\hat{X}_i, \tilde{\theta})=0$, and hence 
 	 	$\frac{1}{\sqrt{n}} \sum_{i=1}^n f_j(X_i, \tilde{\theta})$
 	 	can be made arbitrarily small, with probability close to 1, for $n$ large.
 	 	Furthermore, by definition,
 	 	$\frac{1}{n} \sum_{i=1}^n f_j(X_i, \hat{\theta})=0$.
 	 	Suppose, then, that there exists a positive constant $c_1$ such that for $n$ sufficiently large  $$P_{\theta_0}(|\tilde{\theta}-\hat{\theta}|>c_1/\sqrt{n}) \geq \alpha>0$$
 	 	Expanding the function
 	 	$h_j(\theta)=\frac{1}{n}\sum_{i=1}^n f_j(X_i, \theta)$ in a second-order Taylor expansion around $\hat{\theta}$, we find
 	 	$$h_j(\tilde{\theta})=h_j(\hat{\theta})+ \nabla h_j^{\top}(\hat{\theta})(\tilde{\theta}-\hat{\theta}) + [\tilde{\theta}-\hat{\theta}]^{\top} H(\theta^*) [\tilde{\theta}-\hat{\theta}]$$
 	 	where $H$ is the Hessian matrix of $h_j$ evaluated at some point $\theta^*$ on the line segment between $\hat{\theta}$ and $\tilde{\theta}$. We assume that the Hessian is bounded in spectral norm. From the above equality, we conclude
 	 	$$\sqrt{n} h_j(\tilde{\theta})=0+ \nabla h_j^{\top}(\hat{\theta}) \sqrt{n}(\tilde{\theta}-\hat{\theta})+ \sqrt{n}(\tilde{\theta}-\hat{\theta})^{\top} H(\theta^*) (\tilde{\theta}-\hat{\theta})$$
 	 	Put
 	 	$$v_j(n)=\sqrt{n}(\tilde{\theta}-\hat{\theta})^{\top} H(\theta^*) (\tilde{\theta}-\hat{\theta})$$
 	 	Because of the boundedness of the Hessian and the fact that in probability, $\tilde{\theta}-\hat{\theta} \rightarrow 0$,
 	 	we have
 	 	$$\frac{|v_j(n)|}{||\sqrt{n}(\tilde{\theta}-\hat{\theta})||} \rightarrow 0$$ 
 	 	in probability, so that the error term $v_j$ is of smaller order than the norm of $\sqrt{n}(\tilde{\theta}-\hat{\theta})$.
 	 	Now, consider the vectors 
 	 	 $$\bh=[h_1(\tilde{\theta}), \dots,  h_l(\tilde{\theta})]^{\top} \textrm{ and } \bv=[v_1(n), \dots, v_l(n)]^{\top}$$
 	 	Observe that if we define $\bS$ by $\bS_{bc}(\bX,\theta)=-\frac{1}{n}\sum_{i=1}^n \frac{\partial^2 f(X_i, \theta)}{\partial \theta_b \partial\theta_c}$, for $b,c \in 1, 2, \dots, l$, then
 	 	\begin{equation}\label{eq:penultimate_contradiction}
 	 	\bS^{-1}(\hat{\theta})\sqrt{n} \bh(\tilde{\theta})-\bS^{-1}(\hat{\theta}) \bv=\sqrt{n}(\tilde{\theta}-\hat{\theta})
 	 	\end{equation}
 	 By our functional central limit theorem, the first component on the left hand side of \eqref{eq:penultimate_contradiction}, namely $\bS^{-1}(\hat{\theta})\sqrt{n} \bh(\tilde{\theta})$, goes to zero. The norm of the second component on the left hand side of \eqref{eq:penultimate_contradiction}, namely $\bS^{-1}(\hat{\theta}) \bv$, is of asymptotically smaller order than $\sqrt{n}(\|\tilde{\theta}-\hat{\theta}\|)$.  Hence if, for $n$ sufficiently large,  $$P_{\theta_0}(\sqrt{n}\|(\tilde{\theta}-\hat{\theta}\|)>c) \geq \alpha>0$$
 	 	we obtain a contradiction. Therefore,
 	 	$\sqrt{n}(\tilde{\theta}-\hat{\theta})$ converges to zero in probability. Observe that by Eq.~\eqref{eq:fisher-info-convergence} and the consistency of the maximum likelihood/minimum contrast estimate, we find that $\bS^{-1}(\hat{\theta}) \rightarrow I(\theta_0)$. The result now follows from Slutsky's Theorem. 
 	 	\end{proof}
 	 	
\begin{remark}
We note that the definition of LSM includes the sparsity parameter $\rho_n$. If we let $\rho_n \rightarrow 0$, so that the graph densities decrease as $n$ increases, then Theorem 1 and Theorem 2 need to be adjusted accordingly. For example, if $\rho_n \rightarrow 0$, then the $\sqrt{n}$ scaling in Theorem 2 is replaced by a scaling of $\sqrt{n \rho_n}$; that is, the estimation accuracy of the $\hat{\mathbf{X}}$ decreases as $\rho_n$ decreases. Theorem 4 then needs to be restated, in that the efficiency of the $\{\hat{X}_i\}$ is identical to that of a smaller sub-sample of the $\{X_i\}$; more specifically, we sub-sample $\rho_n^{1/2}$ of the $X_i$ and use them to estimate the parameters $\theta$. This is unavoidable, because sparser graphs contain less signal.
\end{remark}

 	\section{Examples of efficient estimation and testing for latent structure models}\label{sec:Examples}
 	
 	To illustrate the results numerically, we first consider the parametric latent structure model with known support, constructed as follows.  Let $G$ be the cumulative distribution function of the $\textrm{Beta}(a,b)$ distribution, and let $r:[0,1]\rightarrow \bC$ be the map $r(t)=(t^2, 2t(1-t), (1-t)^2)$. Then $\mathcal{C}=\textrm{Im}(r)$ describes the Hardy-Weinberg ($H-W$) curve in the simplex. Let $p$ be the arclength parametrization of $\mathcal{C}$.
 	
 	Consider a random dot product graph with latent position matrix $\bX$ whose rows are i.i.d draws from $F=G(p^{-1})$ along the Hardy-Weinberg curve.  Let $\bA$ be the adjacency matrix of this graph, and let $\hat{\bX}_i$ be the $i$th row of the corresponding adjacency spectral embedding, {\em suitably rotated}. Recall that an appropriate rotation is necessary because of the inherent nonidentifiability in our model. We note that in the simulations we discuss below, we generate the true latent positions first, and as such are able to determine the particular orthogonal transformation that optimally aligns the estimated latent position with the true latent positions. When processing real data, of course, this rotation is unknown. Manifold learning can still be used for the estimation of a rotation of the curve. In two-sample testing, this orthogonal nonidentifiability can be addressed by by determining an optimal Procrustes fit between pairs of point clouds of estimated positions.
    
    For notational simplicity, we continue to refer to $\hat{\bX}$ as the matrix of {\em suitably rotated} latent positions.  Consider Figure \ref{fig:HW2_beta_12}, which illustrates the components of a latent structure model with known structural support and also depicts our methodology for parametric estimation in this context. In panel (a), we see the density $G_{\theta}=\textrm{Beta}(a=1,b=2)$ on the unit interval, and a subsample of points $t_i$, depicted as a rug plot, chosen from this density; this is the underlying distribution for our latent structure model. We do {\em not} observe this distribution. Because we are in a parametric latent structure model, we assume the underlying distribution belongs to a parametric family (in this case the Beta family), but we do not assume knowledge of any or all of the relevant parameters. In panel (b), we see the images of these points, $p(t_i)$, along the Hardy-Weinberg curve; these are the true latent positions that generate our random graph. Once again, we do {\em not} observe these points. In panel (c), we see the random dot product graph generated from these true latent positions. This is the network we actually {\em do} observe. In the panel (d), we see the adjacency matrix for this network. It does not seem obvious that an observation of the network or its adjacency matrix would allow us to accurately estimate the latent positions or the underlying Beta distribution. And yet, in panel (e), we see the estimated latent positions given by the rows of the adjacency spectral embedding for the random graph with the previously-specified true latent positions; these follow the true latent positions in panel (b). Last, in panel (f), we show the $\textrm{Beta}(\tilde{a}, \tilde{b})$ density that arises when computing the $M$-estimates for the parameters $a,b$ based on the estimated latent positions. This final panel shows a striking similarity to panel (a), the true underlying distribution.

\begin{figure}[htp]
 		\centering
 		\subfloat[][{\em Unobserved} Beta (1,2) density and rug plot on unit interval]
        {\includegraphics[width=0.28\textwidth]{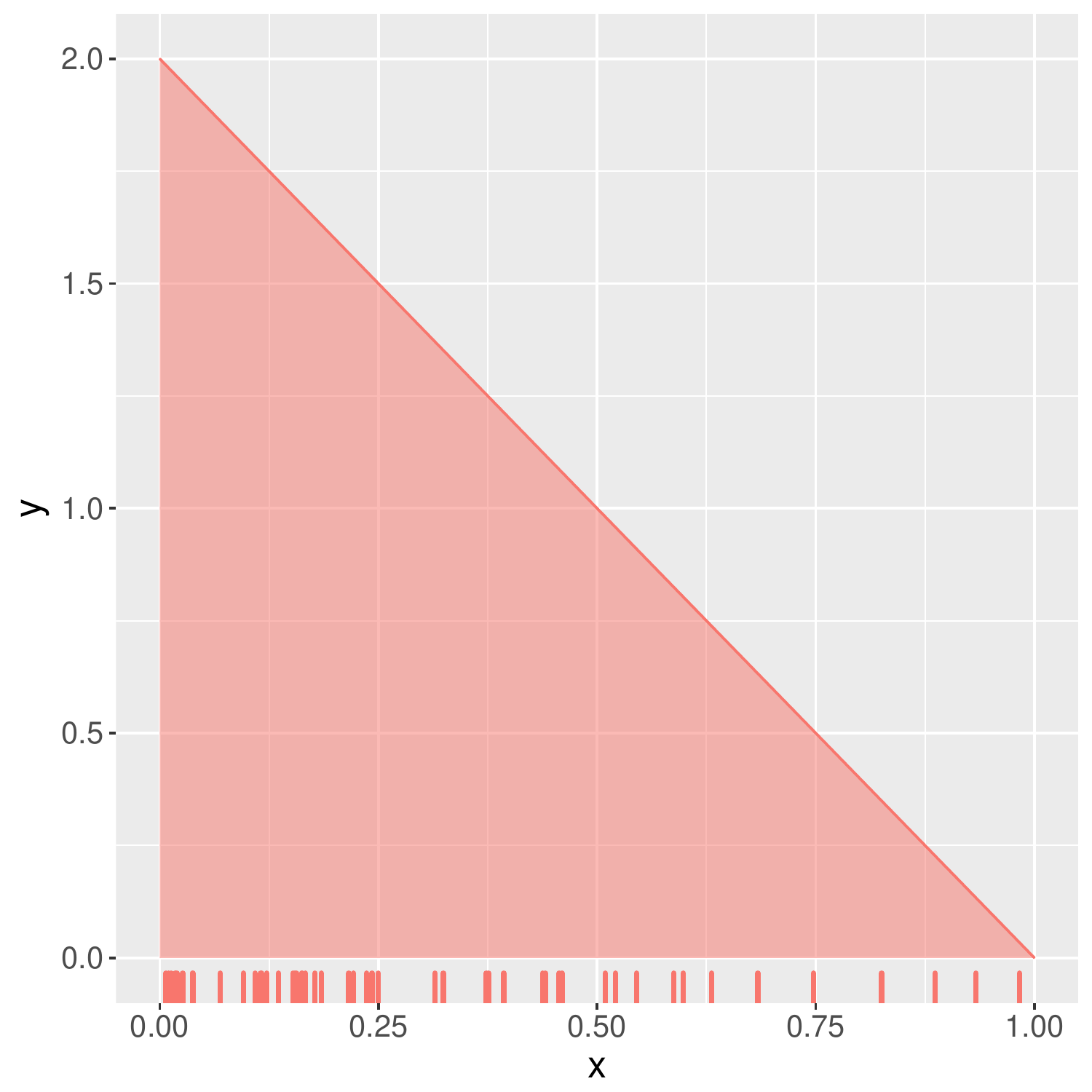}}
        \hfil
 		\subfloat[][{\em Unobserved} latent positions for LSM: images of Beta-distributed points on H-W curve]{\includegraphics[width=0.34\textwidth]{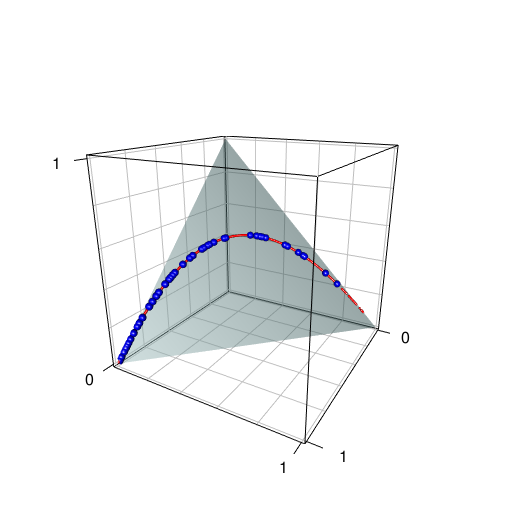}}\\
 		\subfloat[][{\em Observed} random dot product graph generated from these latent positions]
        {\includegraphics[width=0.27\textwidth]{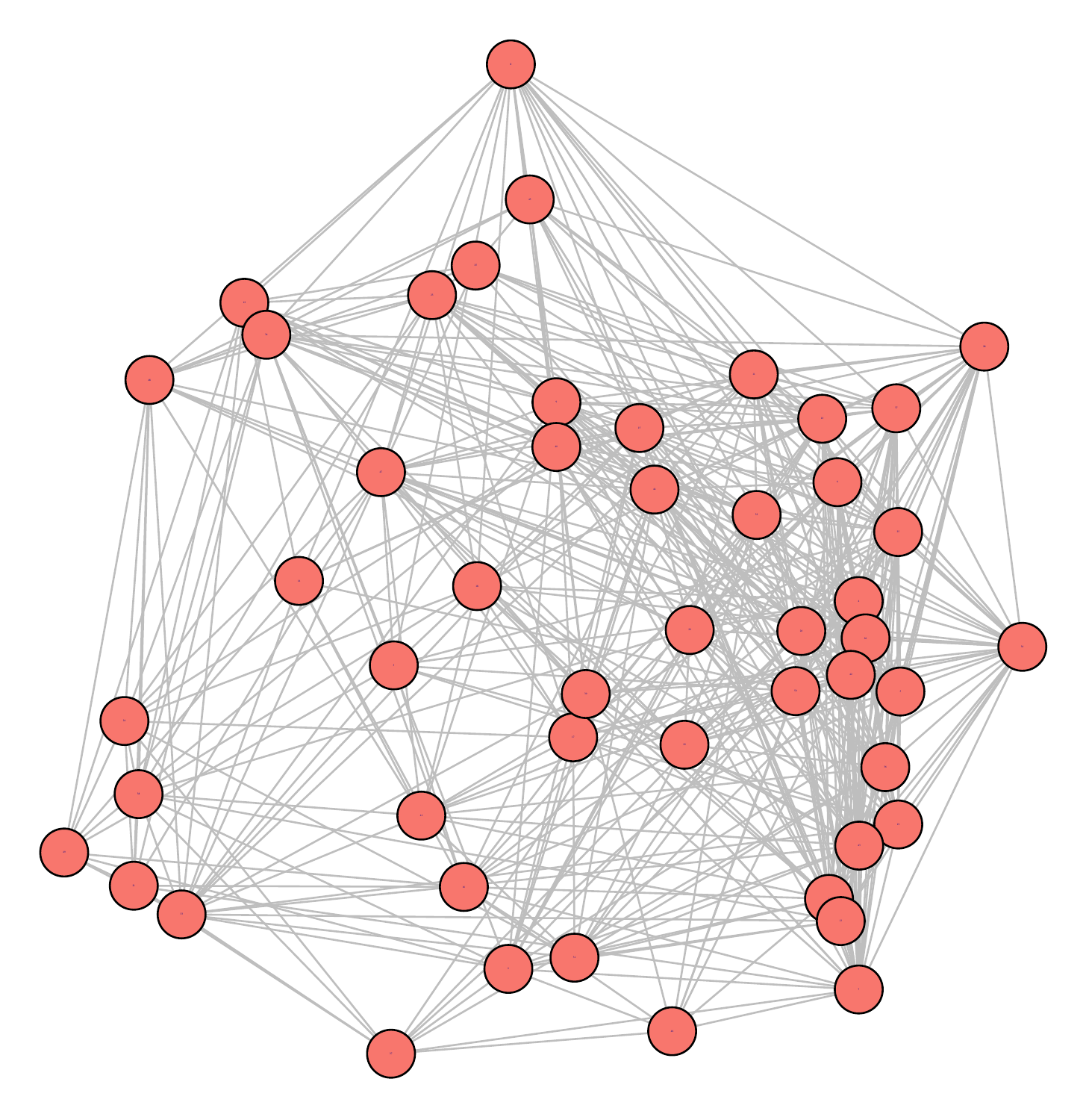}}
        \hfil
 		\subfloat[][{\em Observed} adjacency matrix of the random dot product graph]{\includegraphics[width=0.27\textwidth]{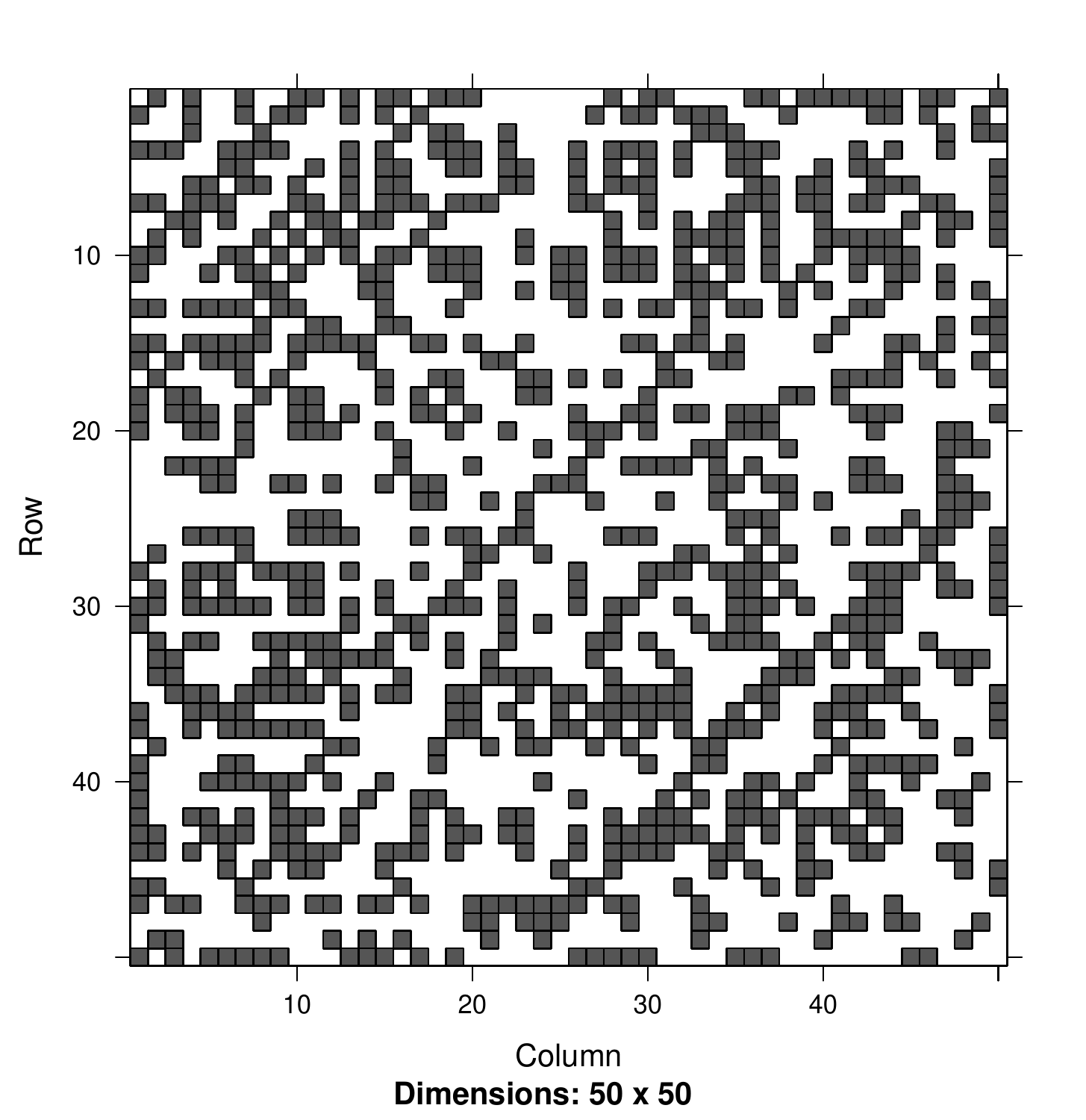}}\\
 		\subfloat[][{\em Estimated} latent positions around H-W curve]{\includegraphics[width=0.34\textwidth]{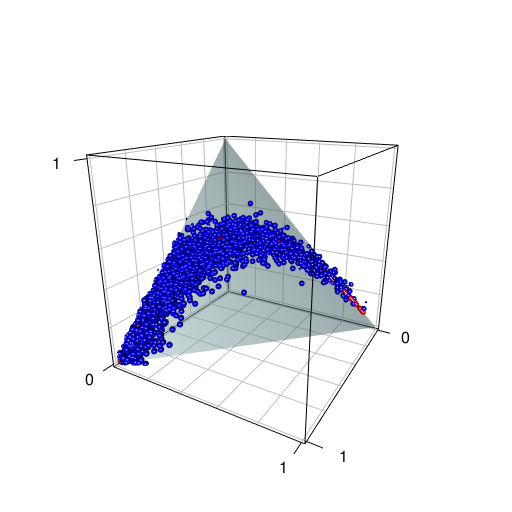}}
        \hfil
        \hfil
 		\subfloat[][{\em M-Estimated} Beta density using estimated latent positions]{\includegraphics[width=0.28\textwidth]{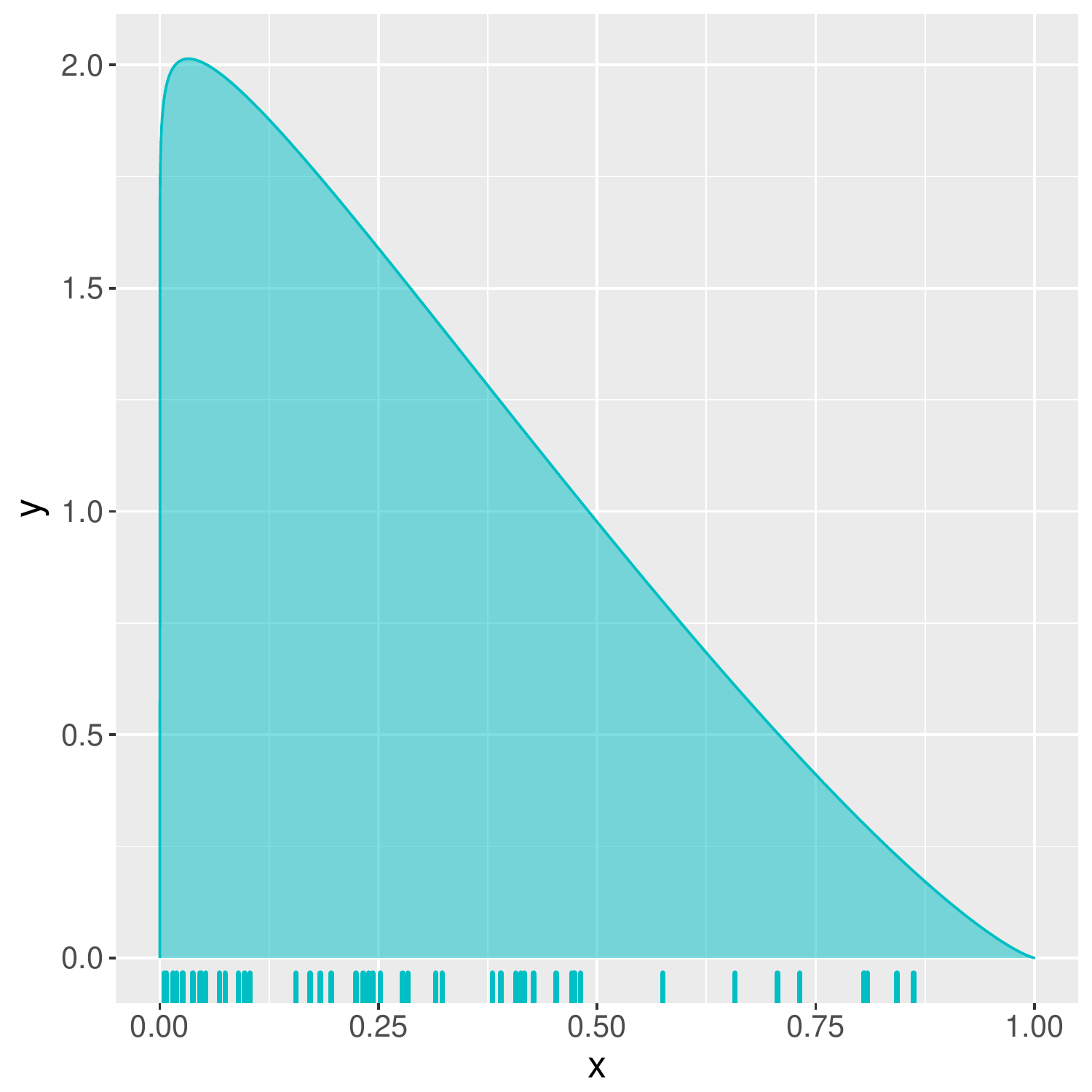}}
 		\caption{(a) Beta $(1,2)$ density (unobserved); (b) latent positions on H-W curve (unobserved); (c) RDPG with these latent positions (observed); (d) adjacency matrix of RDPG (observed); (e) estimated latent positions around H-W curve; (f) $M$-estimated Beta density. See also Table \ref{table:first_MSE}.}
 		\label{fig:HW2_beta_12}
\end{figure}

	 With our Beta parameters $\theta=(a=1,b=2)$, let $\hat{\theta}=(\hat{a},\hat{b})$ be the estimates satisfying
 	\eqref{eq:true_MLE}; that is, the maximum likelihood estimates based on the true latent positions, and let $\tilde{\theta} =(\tilde{a}, \tilde{b})$ be the $M$-estimates satisfying \eqref{eq:theta_tilde}; that is, the quasi-maximum likelihood estimates based on the estimated latent positions.  Table \ref{table:first_MSE} shows the mean-squared error (MSE) for each of these estimates at sample size $n=8000$, demonstrating that these estimates yield comparable mean-squared error for $n=8000$ (see Table \ref{table:8000} for the MSE for other parameter values).
 	\vspace{0.1in}
 		\begin{table}[htp]
 			\caption{Mean-squared error of Beta $a=1, b=2$ parameters\\ in an H-W LSM using true and estimated latent positions\\Sample size $n=8000$}
            \label{table:first_MSE}
 	\begin{tabular}{|c|c|c|}	
 	 \hline
 	 & a & b \\ \hline
 	 & & \\ 
 	{\bf MSE($\hat{\theta}$)} & 0.00014 & 0.00097\\
 	& & \\
 	\hline
 	& & \\
 	{\bf MSE($\tilde{\theta}$)} & 0.00015 & 0.0012\\
 	& & \\
 	\hline
 	\end{tabular}
 	\end{table}
 	
 	This simulation renders plausible our central claim that in a latent structure model, $M$-estimation using the estimated latent positions compares favorably to $M$-estimation using the true latent positions. 
 	
 	We next consider the case when the support is unknown, but parametric.  As before, let $G$ be the cumulative distribution function of the $\textrm{Beta}(a,b)$ distribution on the unit interval, and $t_i \in [0,1]$ a collection of independent, identically $G$-distributed random variables. Let $\mathcal{C}$ be a curve with minimal subspace dimension $d$.  Suppose that $\mathcal{C}$ is the image of a map $
    q:[0,1] \rightarrow \mathbb{R}^d$ where each component $q_k(\cdot)$ of $q$ is a polynomial of some fixed degree (for example, quadratic). Once again, let $p$ represent the arclength parametrization of $\textrm{Im}(q)$. Consider a latent structure random graph with adjacency matrix $\bA$ whose latent positions are given by $X_i \in \mathcal{C}$, where, as before, $X_i$ are i.i.d $F=G(p^{-1})$.  In this case, we have two separate estimation problems before us: an estimation of the parameters defining each quadratic polynomial $q_k(t)$---or, equivalently, an estimation of the curve $\mathcal{C}$; and second, an estimation for the parameters $a, b$.
 	
 	Considered individually, neither of these is insurmountable: if we have enough i.i.d data centered along a polynomial curve, we can estimate the curve. Similarly, given enough i.i.d draws of points in the interval, we can estimate the parameters of our $\textrm{Beta}$ distribution. But in our setting, we have only {\em non-i.i.d data} around an {\em unknown} curve. Thus even if we could reasonably use the estimated latent positions to recover the structure of the support of our distribution $F$, there remains the recovery of parameters in a wholly different space.
 	
 	The efficiency of $M$-estimation composed with the adjacency spectral embedding (in particular, as discussed in Sec. \ref{sec:asymp_eff_LSM}, the consistency result of  Theorem~\ref{thm:minh_sparsity} and the uniform convergence result of Theorem~\ref{thm:u-statistics}) allows us to connect these two inference procedures.  In the figures below, we once again consider simulated data along the Hardy-Weinberg curve.  Instead of assuming knowledge of the precise map $p$ defining this curve, however, we assume only that it is quadratic, and we attempt to learn the parameters of this quadratic curve from the estimated latent positions---that is, from the adjacency spectral embedding of the latent structure random graph.

 	In particular, consider Figure \ref{fig:HW_alpha_beta_many}. Each panel in Fig. \ref{fig:HW_alpha_beta_many} shows a two-dimensional projection (on to the first two coordinates) of estimated latent positions drawn from the Hardy-Weinberg curve.  That is, we first simulate 8000 points from a Beta distribution with various parameters: $(a=1, b=1)$; $(a=1, b=2)$; $(a=2, b=5)$; and $(a=5, b=5)$. We consider the images of these points under $p:[0,1] \rightarrow \mathcal{C}$, where $p$ is the arclength parametrization of $\mathcal{C}$, the Hardy-Weinberg curve.  These are the true latent positions $X_i$. We generate a random dot product graph $\bA$ with these latent positions, and then spectrally embed $\bA$ into $d=3$ dimensions.  Fig. \ref{fig:HW_alpha_beta_many} shows scatter plots of the first two coordinates of the estimated latent positions; these are the blue dots around the black Hardy-Weinberg curve $\mathcal{C}$ on which the true latent positions lie. We use these estimated latent positions to obtain a best-fitting quadratic Bezier curve $\hat{\bC}$ \cite{gallier,prautzsch} through these positions, shown in red. The quadratic restriction implies that estimating the structural support of our latent structure model reduces to estimating three $3$-dimensional parameters, so that we can reduce a nonparametric problem of curve-fitting to a parametric problem of the estimation of coefficients of a quadratic.  Nevertheless, as Fig. \ref{fig:HW_alpha_beta_many} shows, the accuracy of the estimation of the support can depend considerably on the $\textrm{Beta}$ parameters themselves; in the uniform $(a=1, b=1)$ case, the estimated Bezier curve tracks the true Hardy-Weinberg curve nicely. At $(a=1, b=2)$, we retain most of this accuracy.
At both $(2,5)$ and $(5,5)$, we see a marked deviation between the estimated curve (in black) and the true Hardy-Weinberg curve (in red). In particular, as the parameters $(a,b)$ change, the points of the Beta distribution can cluster around a central mode or, alternatively, tend to drift further apart, toward the endpoints of the interval. These alterations in the shape of the underlying distribution can lead to a poor estimate of the parameters of the best-fitting quadratic through the estimated points, resulting, in turn, in a less-accurate estimate of the LSM's support.
\begin{figure}[ht]
\centering
 \subfloat[][$a=1,b=1$]{
   \includegraphics[width=0.45\textwidth]{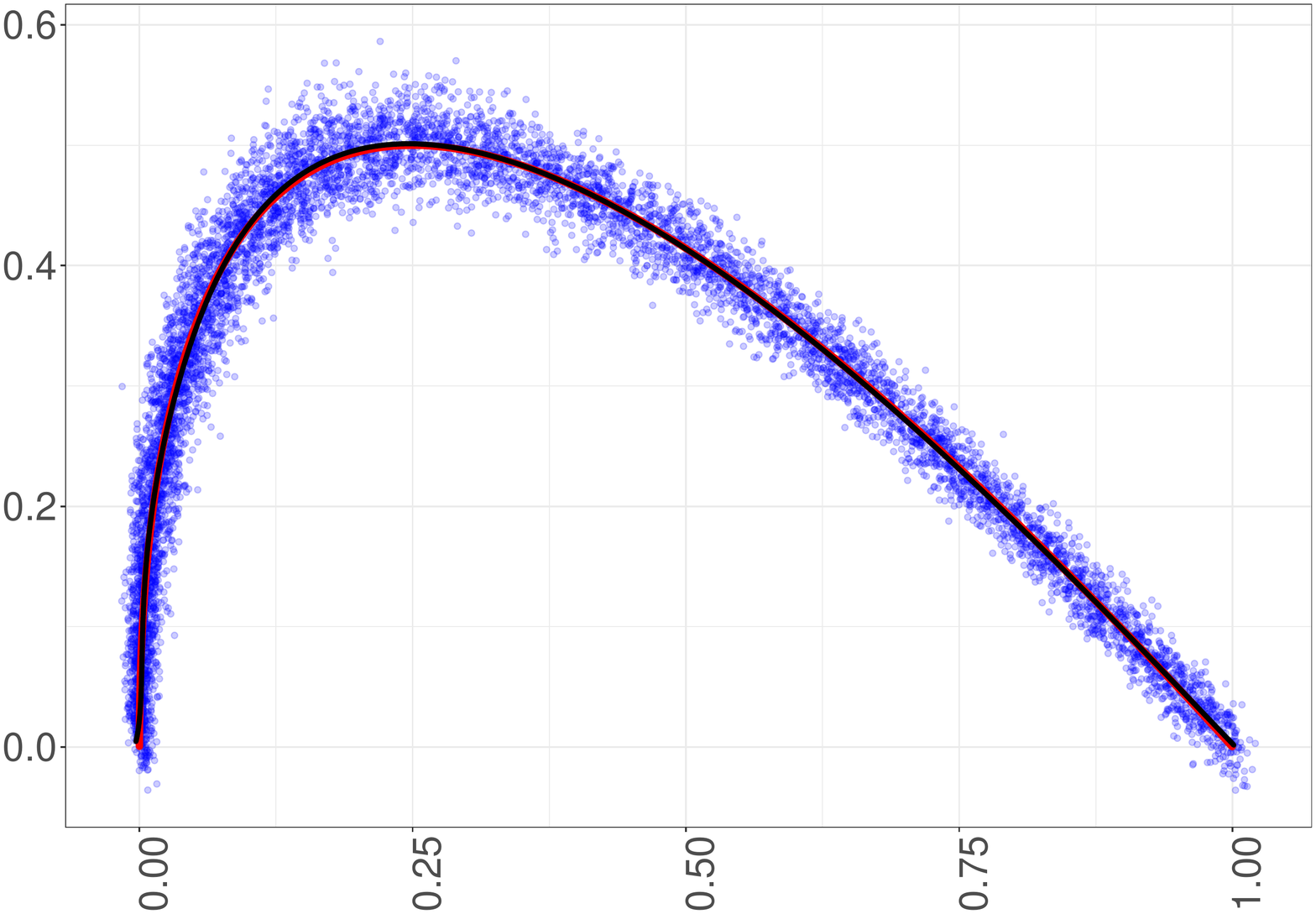}
 } \hfil
 \subfloat[][$a=1,b=2$]{
   \includegraphics[width=0.45\textwidth]{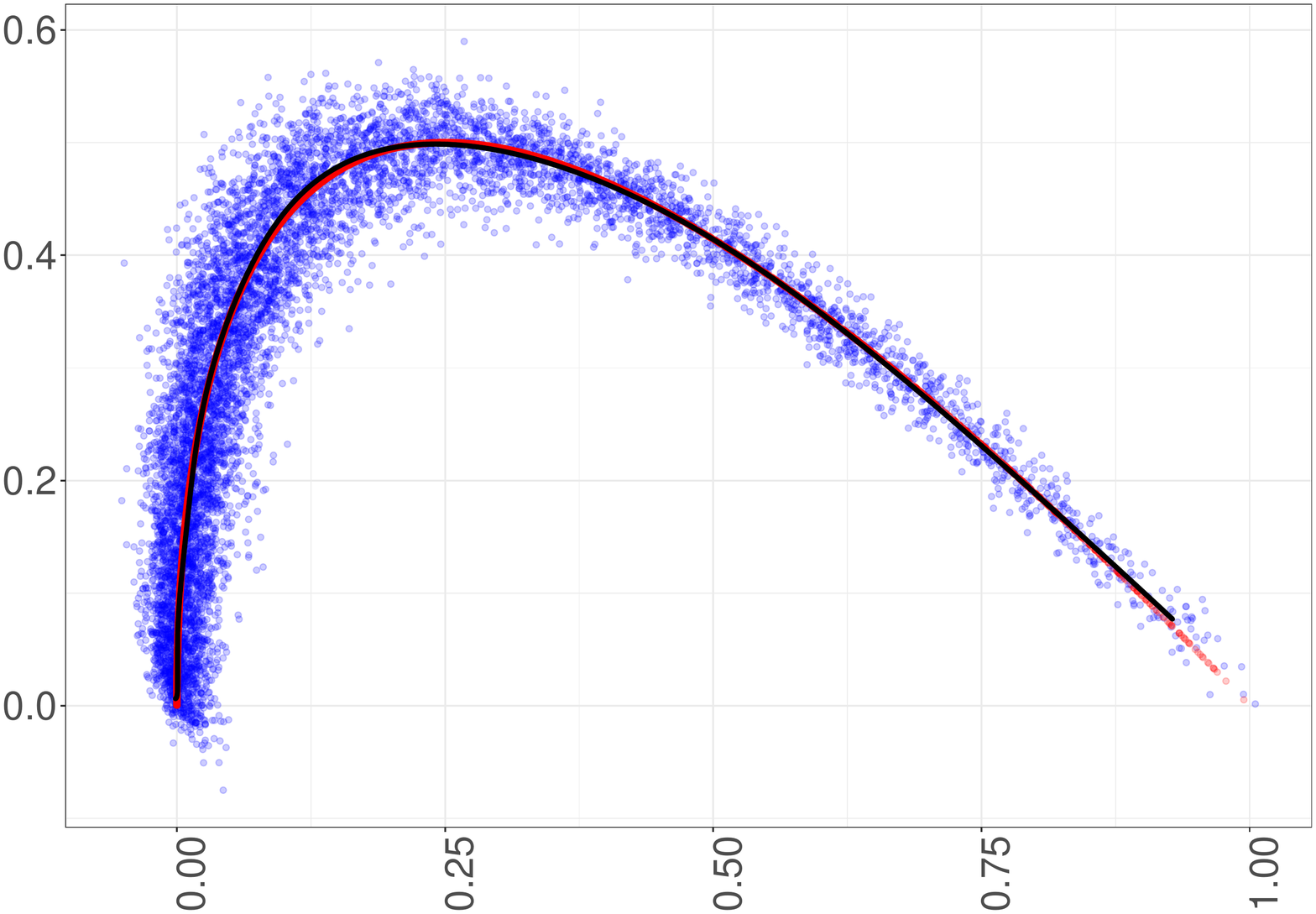}
 } \\
  \subfloat[][$a=2,b=5$]{
   \includegraphics[width=0.45\textwidth]{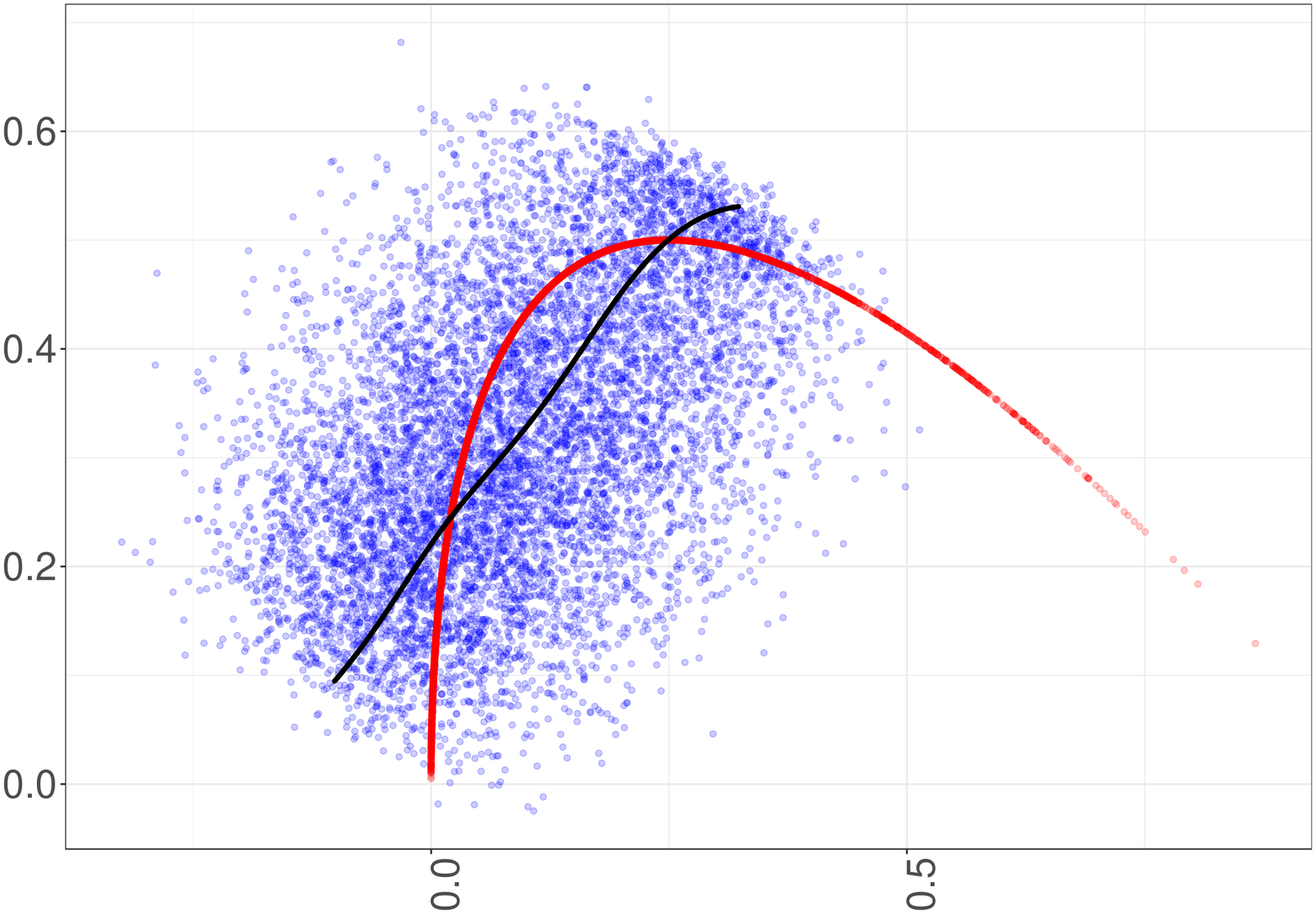}
 } \hfil
 \subfloat[][$a=5,b=5$]{
   \includegraphics[width=0.45\textwidth]{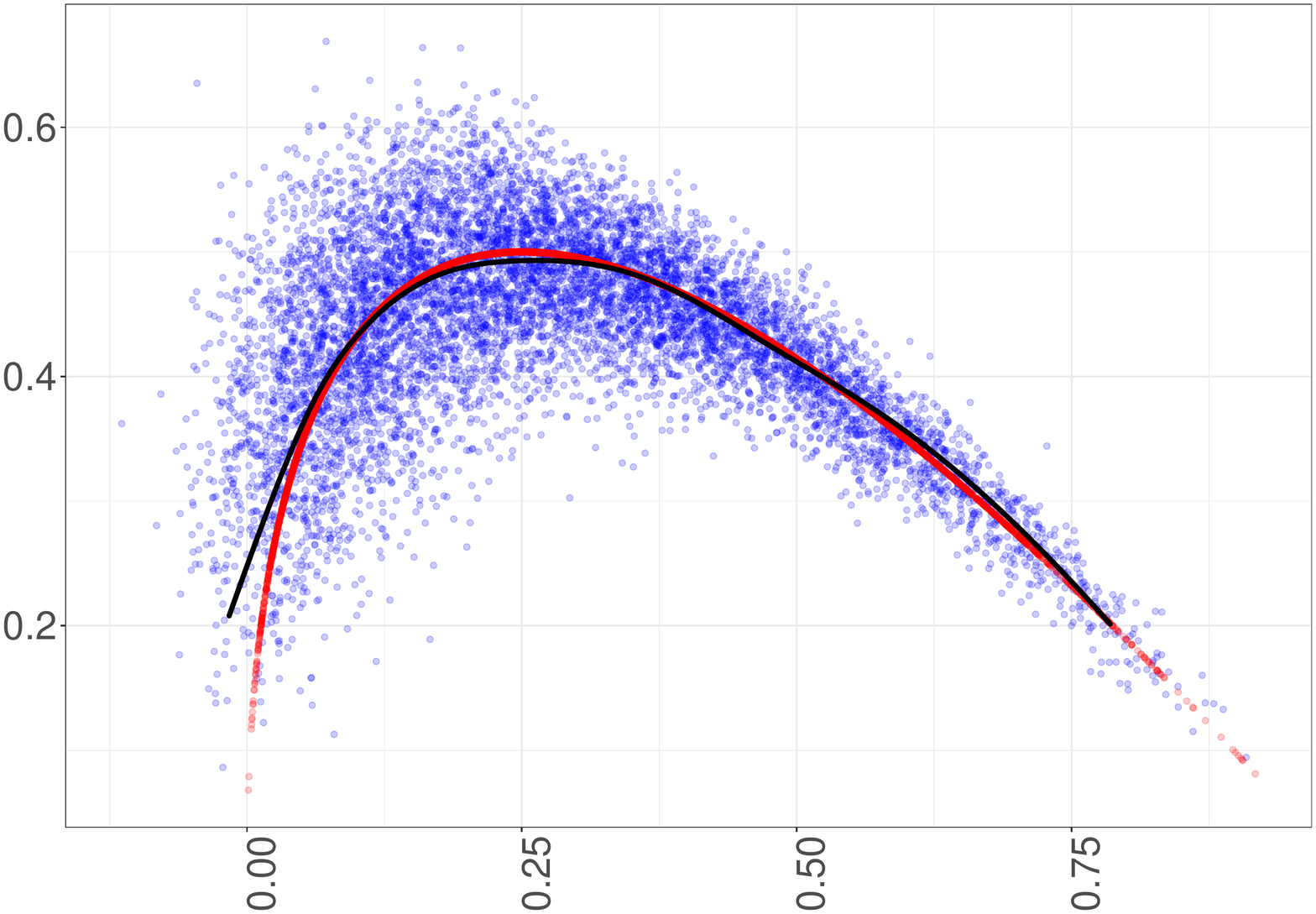}
 }
\caption{Bezier curve estimates through the estimated latent positions for underlying distributions $\textrm{Beta}(1,1)$ (top left); $\textrm{Beta}(1,2)$ (top right); $\textrm{Beta}(2,5)$ (bottom left); and $\textrm{Beta}(5,5)$ (bottom right).}
	\label{fig:HW_alpha_beta_many}
\end{figure}

Finally, we consider the projections onto the estimated Bezier curve $\hat{\mathcal{C}}$ of each latent position, and then the inverse images of these projected points $p^{-1}(\pi_{\hat{\mathcal{C}}}(\hat{X}_i))$ in the unit interval.  To obtain estimates for the Beta parameters, we consider the $M$-estimate defined by
\begin{equation}\label{eq:M-est_whole_enchilada}
\hat{a}, \hat{b}=\arg\max_{a,b} \sum_{i=1}^n \log g_{a,b}(p^{-1}(\pi_{\hat{\mathcal{C}}} (\hat{X}_i)))
\end{equation}
where $g_{a,b}$ is the $\textrm{Beta}(a,b)$ density. Eq. \eqref{eq:M-est_whole_enchilada} is a single-line summary of our entire methodology. In a latent structure model random graph with unknown but parametrically-determined support, we observe the adjacency matrix of the graph; compute the adjacency spectral embedding to yield the estimated latent positions $\hat{X}_i$; use the consistency of these estimates for the true latent positions $X_i$ to obtain an accurate estimate $\hat{\mathcal{C}}$ of the true structural support $\mathcal{C}$; then use the pullbacks of the projections $\pi_{\hat{\mathcal{C}}}(\hat{X}_i)$ as inputs into an $M$-estimate for the parameters of our underlying distribution $G_{\theta}$. Because there is now distortion through curve estimation, projection and pullback, we conjecture that consistency, at the parametric rate, is achievable here, but that asymptotic efficiency may well be lost.

%
%
%
%

Morever, there are two particular issues with Eq.~\eqref{eq:M-est_whole_enchilada} that bear noting.  The first is that any individual point $p^{-1}(\pi_{\hat{\bC}}(\hat{X}_i))$ may correspond to an endpoint of the unit interval, at which the underlying density may be zero (as in the Beta case). To avoid this numerical artifact, we scale these points slightly, by an infinitesimal $\epsilon>0$.  We underscore again that the consistency of the latent position estimates \cite{lyzinski15_HSBM} implies that, as the sample size $n$ grows, this adjustment affects an increasingly smaller fraction of estimated latent positions, and thus does not impact our limiting results. But it does render necessary certain adjustments in the finite-sample case. The second issue, as we mentioned earlier, is that the while the Bezier curve can be accurately estimated in the {\em limit} for any $a, b$, the finite-sample case is trickier, and the error inherent in the estimation of the support can have unpleasant downstream consequences for the estimation of the underlying parameters.

%
%
%
Predictably, the deteriorating quality of the Bezier curve estimate also impacts the mean squared error of our $M$-estimates for $a$ and $b$. In Tables \ref{table:1000} and \ref{table:8000}, we present the MSE of our latent position estimates for $(a=1, b=1)$; $(a=1, b=2)$; $(a=2, b=5)$; $(a=2,b=2)$; $(a=5,b=5)$, with the sample sizes of $n=1000$ and $n=8000$. Note the sharp contrast of the MSE for the $(a=5, b=5)$ case, which reflects the challenge of estimating the support even when it is a parametrically specified curve. The impact of the parameters of the underlying distribution $G_{\theta}$ on subsequent inference is a topic of current work. Indeed, the distressingly large mean-squared error for certain values of $(a, b)$ highlights the utility, in theory and practice, that one could derive from a second-order Berry-Esseen result describing precisely how robust this procedure is to values of these parameters.


\begin{table}[htp]
	\caption{Mean-squared error of Beta parameters in an H-W LSM. \\Sample size $n=1000$}
	\label{table:1000}
	\begin{tabular}{|c|c|c|c|c|c|}	
		\hline
		$\mathrm{MSE}$ & $a=1,b=1$ & $a=1,b=2$ & $a=2,b=5$ & $a=2,b=2$ & $a=5,b=5$ \\ \hline 
		$\mathbf{X}$ & $(0.0061,0.0051)$ & $(0.00068,0.0028)$& $(0.0044,0.039)$ & $(0.0089,0.011)$ & $(0.051,0.051)$ \\ \hline
		$\hat{\mathbf{X}}$ (inverse HW) & $(0.006,0.005)$ & $(0.004,0.019)$ & $(0.4, 2.68)$ & $(0.055,0.033)$ &$(1.14,0.99)$ \\ \hline
		$\hat{\mathbf{X}}$ (inverse Bezier) & $(0.019, 0.02)$ & $(0.08, 0.91)$ & $(1.1, 13.52)$ & $(0.796,0.836)$ & $(14.15, 14.18)$ \\ \hline 
		$\text{RE}(\text{true} \mathbf{X}, \text{inverse HW})$ & $(1,1)$ & $(5.9, 6.8)$ & $(90,68.7)$ & $(6.2, 3)$ &$(22.4,19.4)$ \\ \hline
		$\text{RE}(\text{true} \mathbf{X}, \text{inverse Bezier})$ & $(3.1,3.9)$ &$(117,6.3)$ &$(250,346)$ &$(89,76)$ &$(277,278)$ \\ \hline
	\end{tabular}
\end{table}

\begin{table}[htp]
	\caption{Mean-squared error of Beta parameters in an H-W LSM. \\Sample size $n=8000$}
	\label{table:8000}
	\begin{tabular}{|c|c|c|c|c|c|}	
		\hline
		$\mathrm{MSE}$ & $a=1,b=1$ & $a=1,b=2$ & $a=2,b=5$ & $a=2,b=2$ &$a=5,b=5$ \\ \hline 
		$\mathbf{X}$ & $(0.00015, 0.000083)$ & $(0.00014, 0.00097)$& $(0.0013,0.0098)$ & $(0.0008, 0.0007)$ &$(0.0062,0.0039)$ \\ \hline
		$\hat{\mathbf{X}}$ (inverse HW) & $(0.00023,0.000097)$ & $(0.00015,0.0012)$ & $(0.19, 1.04)$ & $(0.0013, 0.0013)$ & $(0.19,0.13)$ \\ \hline
		$\hat{\mathbf{X}}$ (inverse Bezier) & $(0.0011, 0.0011)$ & $(0.01, 0.14)$ & $(0.61, 10.92)$ & $(0.267,0.267)$ & $(11.55, 11.48)$ \\ \hline
	   $\text{RE}(\text{true} \mathbf{X}, \text{inverse HW})$ & $(1,5,1.1)$ & $(1, 1.2)$ & $(146,106)$ & $(1.6, 1.8)$ &$(30.6,33.3)$ \\ \hline
		$\text{RE}(\text{true} \mathbf{X}, \text{inverse Bezier})$ & $(7.3,13.2)$ &$(71,144)$ &$(469,1114)$ &$(333,381)$ &$(1862,2943)$ \\ \hline
	\end{tabular}
\end{table}

We have, thus far, focused on numerical estimation for latent structure models when (i) the support is known and the underlying distribution is parametric, and (ii) the support is unknown but parametrically specified, and the underlying distribution is parametric. In Tables \ref{table:1000} and \ref{table:8000}, we see that for $\textrm{Beta}(2,5)$, questionable quality of the Bezier curve estimate can have negative consequences for subsequent inference. As a transition to the case of nonparametric estimation for the structural support curve, we consider a two-sample test in the Hardy-Weinberg case. Let $\bA_1$ and $\bA_2$ be two independent adjacency matrices for a pair of latent structure models, both with underlying distribution $\textrm{Beta}(2,5)$. Let $\hat{\bX}_1$ and $\hat{\bX}_2$ be the associated adjacency spectral embeddings. Instead of curve-fitting, we use \texttt{isomap} \cite{isomap_science} to estimate inter-point geodesic distances between the projections of the estimated points $\hat{X}_{1,i}$ (the $i$th row of $\hat{\bX}$, for $1 \leq i \leq n$) onto the unknown curve, and we scale these inter-point distances to the unit interval. We repeat this process with the $\hat{X}_{2,i}$ points.  Thus we now have two sets of points in the unit interval, and we conduct a Kolmogorov-Smirnov test of equality of distribution. For the alternative, we consider the case when one of the graphs is generated by $\textrm{Beta}(2,5)$ and the other by $\textrm{Beta}(3,4)$.  Again, we use \texttt{isomap} to estimate the inter-point geodesic distances between the projections of the estimated latent positions, and conduct the same Kolmogorov-Smirnov test. We find that the $p$-values in the case of the alternative (unequal distributions) are stochastically smaller than the $p$-values under the null. This illustrates that even when parametric curve estimation goes awry, a nonparametric procedure can still be feasible for some subsequent inference tasks.

This leads us directly to our last example, noteworthy because it provides a statistically principled resolution to an important open question in neuroscience. This is an illustration of the utility of the latent structure model for estimation and subsequent inference, even when the structural support and underlying distribution are neither known nor parametrically specified.

To situate this in context, we summarize material described in far more extensive detail in \cite{MBStructure} (and encapsulated again in \cite{athreya_survey}). In particular, recent developments in neuroscience and imaging technology have rendered possible the full mapping of the {\em Mushroom Body} connectome of the larval {\em Drosophila} brain (see \cite{EichlerSubmitted}), which consists of four distinct neuron types--- Kenyon Cells (KC), Input Neurons (MBIN), Output Neurons (MBON), Projection Neurons (PN)---and two distinct hemispheres (right and left).  This connectome can be condensed into a weighted, directed adjacency matrix, specifying which neurons in the mushroom body are synaptically connected to which other neurons.

Our spectral embedding procedure can be adapted for this weighted, directed adjacency matrix, and a suitable embedding dimension can be estimated from the data (again, see \cite{MBStructure} for full details on the spectral decomposition and dimension estimation herein). In order to discern potential differences across the right and left brains, we separately embed the left- and right-hemisphere subgraphs. Neuroscientists conjecture that the right and left hemispheres are {\em bilaterally homologous}--that is,  ``structurally similar." But prior to the formal elucidation of a latent structure model, it was difficult even to {\em frame} this question as a suitable test of hypothesis, let alone provide a principled resolution to it.

However, by considering the mushroom body connectome as a latent structure model with nonparametric structural support and nonparametric underlying distribution, such a hypothesis test becomes both straightforward to construct and feasible to implement. We focus on the estimated latent positions corresponding to the KC neurons and we once again use \texttt{isomap} to learn the structure of the associated support {\em nonparametrically}. As before, \texttt{isomap} returns inter-point geodesic distances between projections of the estimated latent positions; we scale these to yield points $\hat{Y}_i$ in the unit interval, and then feed $\hat{Y}_i$ into a Kolmogorov-Smirnov test for equality of underlying distributions for the right and left KC neurons.

Figure \ref{fig:MBStructure_isomap_1} represents this visually.  Because \texttt{isomap} estimates inter-point distances and not the unknown support curve in $\mathbb{R}^d$, with $d=6$, we provide a two-dimensional visualization of this estimated curve (shown below in red and described in detail in \cite{MBStructure}), representing the structural support of the KC neurons.  The top panels of Fig. \ref{fig:MBStructure_isomap_1} show the estimated latent positions for both the left and right hemisphere KC neurons, as well as a two-dimensional version of the estimated support curve for the right hemisphere alone, which indicates that the support curve for the right hemisphere fits well the data for the left hemisphere. The central panels of Fig. \ref{fig:MBStructure_isomap_1} show the projections of the estimated latent positions for each hemisphere onto the appropriate estimated support curve for that particular hemisphere. Note that these panels are a two-dimensional representation---the actual projections are in $\mathbb{R}^6$, not in $\mathbb{R}^2$. The bottom panels of Fig. \ref{fig:MBStructure_isomap_1} supply a kernel density estimate for the underlying distribution $G$ of the latent structure model for each hemisphere.
\begin{figure}[htp]
	\centering
	\subfloat[][Estimated support of KC neurons, with right hemisphere estimated support curve in both]{
		\includegraphics[width=0.7\textwidth]{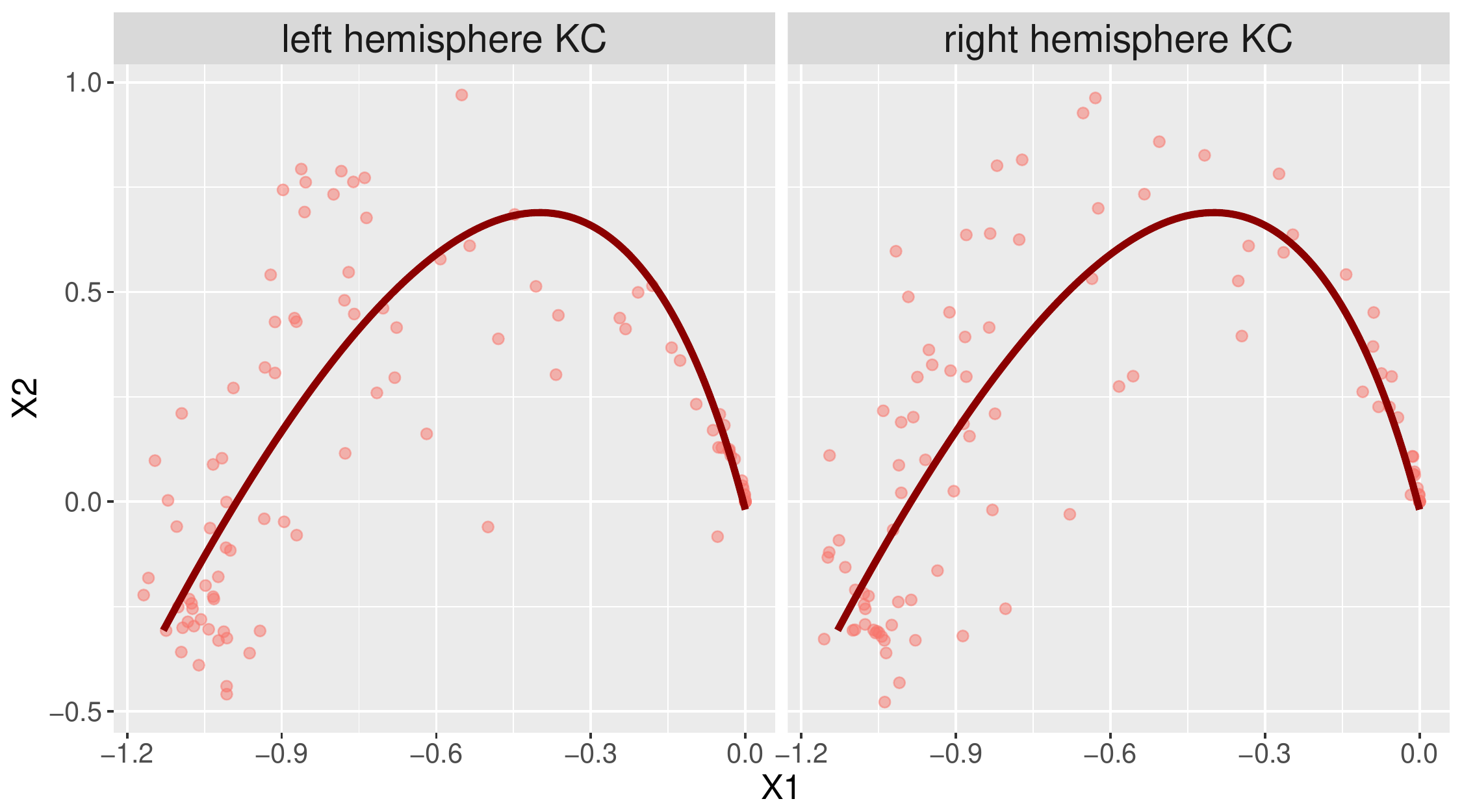}
	} \\
	\subfloat[][Projection of estimated latent positions on to estimated support for each hemisphere]{
		\includegraphics[width=0.7\textwidth]{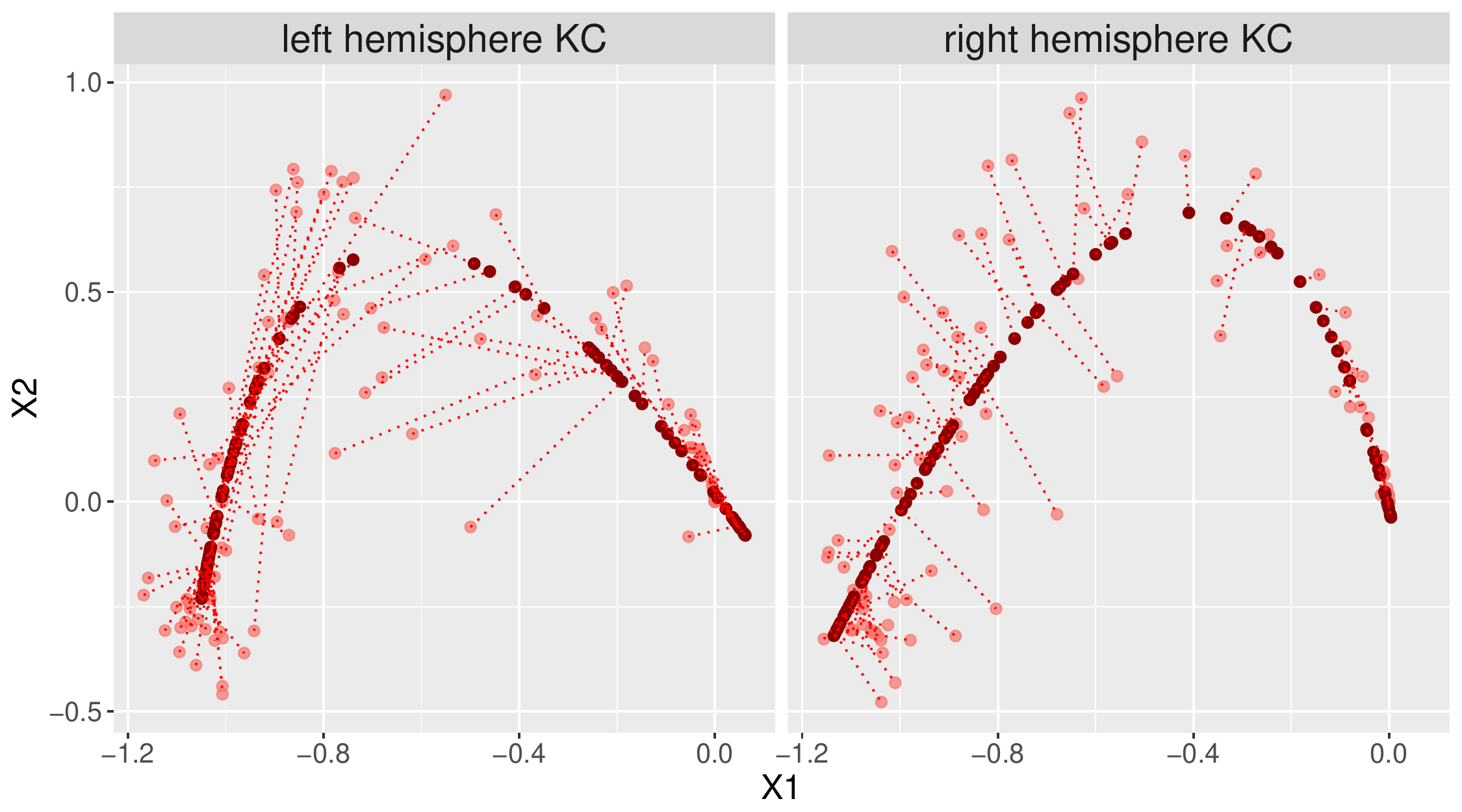}
	} \\
	\subfloat[][Estimated density for underlying distribution for each hemisphere]{
		\includegraphics[width=0.7\textwidth]{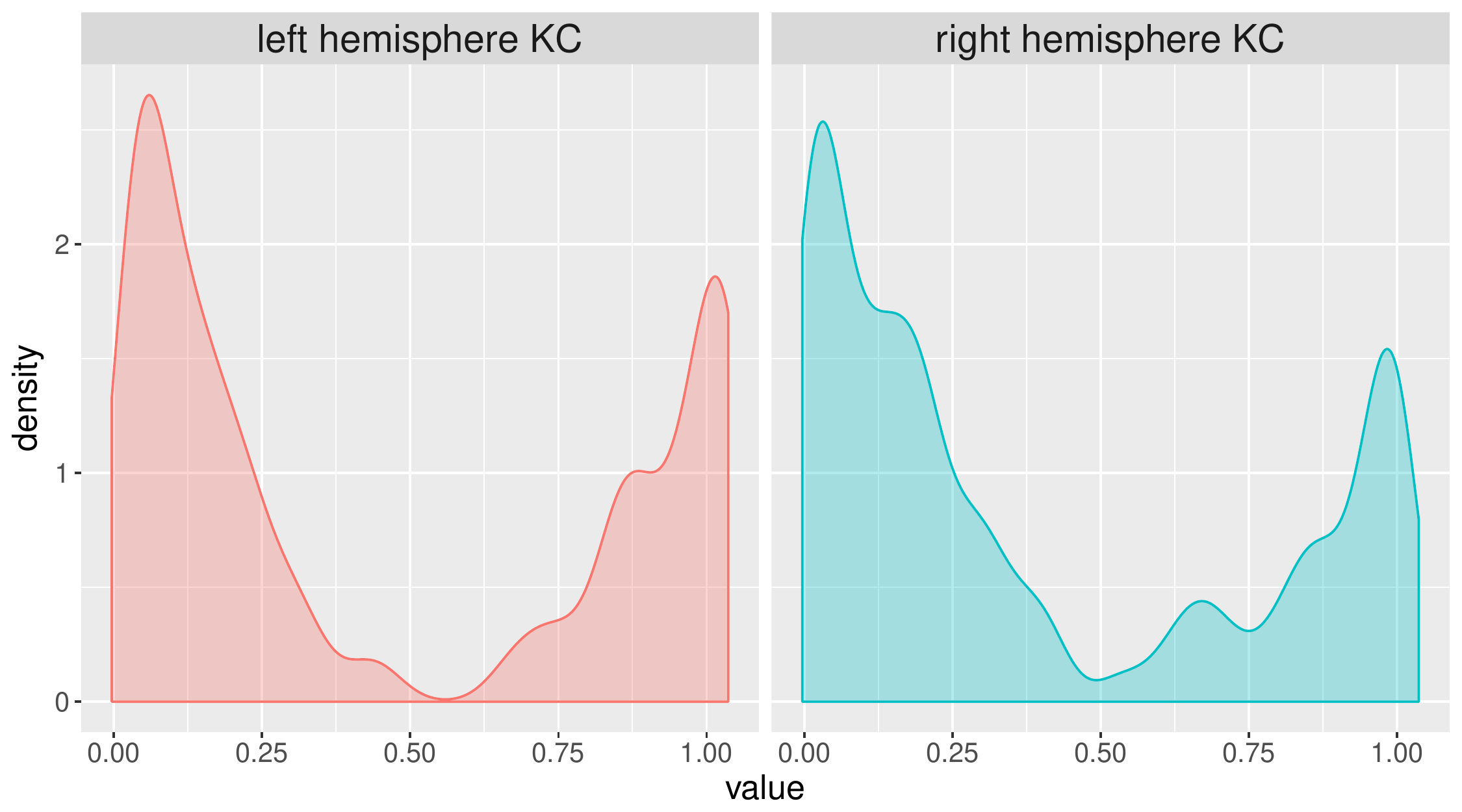}
	}\\
	\caption{Two-dimensional representation of the estimated structural support, projection onto this estimated support for estimated latent positions for KC neurons in the Mushroom Body connectome, and density estimates for the underlying distribution. }
	\label{fig:MBStructure_isomap_1}
\end{figure}

We denote by $\hat{Y}_i^R$ and $\hat{Y}_i^L$ the two sets of scaled inter-point distances obtained via \texttt{isomap} from the estimated latent positions for the right and left hemisphere KC neurons, respectively. Using these as inputs for a Kolmogorov-Smirnov test, we find that we do not reject the null hypothesis of equality of distribution---in fact, we obtain a $p$-value here of $0.68$. Moreover, Fig. \ref{fig:MBStructure_isomap_1} depicts this evidence in favor of our failure to reject the null: a pair of quite similar density estimates for the underlying distribution of the latent structure model for right and left hemispheres. We emphasize, though, that the latent position distribution is not invariant to a reparametrization of the structural support curve under the transformation $t \mapsto 1-t$ of the unit interval. Indeed, if one of the sets of projected points in the middle panels of Fig. \ref{fig:MBStructure_isomap_1} were so reparametrized, this structural symmetry would be destroyed. We find, rather encouragingly, that when we reparametrize one support curve under this transformation, our $p$-value drops to essentially zero. This sensitivity to orientation allows us to rule out the possibility of an underlying uniform distribution---indeed, any underlying symmetric distribution---for the KC neurons. We conclude that we do not reject bilateral homology for KC neurons in the mushroom body, but we do reject the hypothesis of uniformity of underlying distribution. 

The implications of this Kolmogorov-Smirnov $p$-value merit deeper study. For a known curve, the asymptotic validity of this $p$-value
         follows from our earlier results. A learnt curve is a different matter, however.  Nevertheless, we present in Fig. \ref{fig:KS-pval-close-uniform} simulation evidence to suggest that such $p$-values behave nearly uniformly under the null. For this figure, under the null hypothesis of equality of underlying distributions for points on the Hardy-Weinberg curve, we generate a pair of LSM graphs whose latent positions arise as the images on the Hardy-Weinberg curve of a collection of $n=500$ i.i.d Beta $(a=2, b=5)$ points in the unit interval.  We then use the adjacency spectral embeddings for each of the two graphs as inputs into \texttt{isomap}; these inter-point geodesic distances are scaled to land in the unit interval. We finally conduct a Kolmogorov-Smirnov test on these two sets of points. In the alternative, we consider one graph to be generated with Beta ($a=2, b=5$), and the other with Beta ($a=3, b=4$). A rigorous analysis of $p$-value behavior in this type of  Kolmogorov-Smirnov test is the subject of ongoing investigation. 
\begin{figure}[htp]
	\centering
	\includegraphics[width=0.5\textwidth]{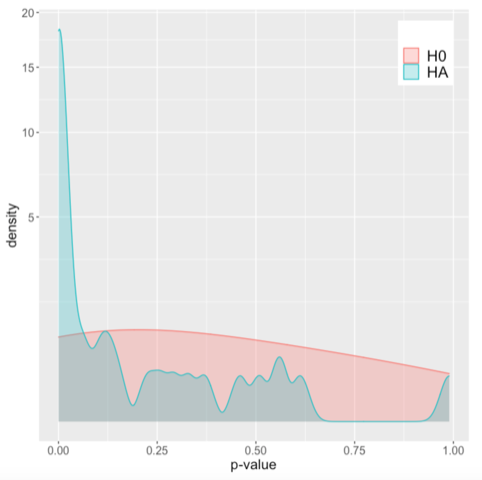}
	\caption{$p$-value distributions for the Kolmogorov-Smirnov test under null ($H_0$) and alternative ($H_A$)}
	\label{fig:KS-pval-close-uniform}
	\end{figure}

 	\section{Conclusion}\label{sec:Conclusion}
 	
 In closing, the latent structure model formalizes an intuitive premise that many random networks have both probabilistic and geometric structure. Defining latent structure models within the class of random dot product graphs provides the advantage of rendering these two components distinctly. First, we specify the probabilistic component of a one-dimensional LSM with an {\em underlying distribution} $G$ on the unit interval, and second, we delineate the geometric structure of the graph by specifying a curve $\mathcal{C}$ as the {\em structural support} of the latent position distribution $F$ of the random dot product graph, where $F$ is defined by $F=\mu_G(p^{-1})$ and $p:[0,1] \mapsto \bC$ is the arclength parametrization of $\mathcal{C}$ on the unit interval. The map $p$ is the connecting thread between points $t_i$ generated in $[0,1]$ according to $G$ and the corresponding images $X_i=p(t_i)$ that are the latent positions for the random graph. Because one-dimensional latent structure models are random dot product graphs that also depend on the parametrization $p$ of the structural support curve $\mathcal{C}$, they have two nonidentifiabilities. One is an immediate consequence of the invariance of the random dot product graph to orthogonal transformations of the latent positions. The second nonidentifiability arises from the fact that the map $p$ can be orientation-reversed by considering $t \mapsto 1-t$.    
 
 Framed as random dot product graphs with structural constraints, latent structure models provide an intermediate point between simple stochastic block models and more general, unconstrained random dot product graphs. Furthermore, by separating the probabilistic and geometric sources of network regularity, we can construct latent structure models according to natural demarcations of increasing probabilistic or geometric complexity: the underlying distributions on the unit interval can be known, parametrically specified, or nonparametric; and similarly the structural support curves $\mathcal{C}$ can be known, parametrically specified, or nonparametric.   
 
 Our main result here is that, to perform efficient estimation of the parameters of a parametric latent structure model with known support, one needs only the estimated latent positions $\hat{X}_i$ arising from a spectral decomposition of an adjacency matrix generated by the true latent positions $X_i$. One does not need to observe the true latent positions themselves.  The efficiency of $M$-estimation via the adjacency spectral embedding is part of a broader program in which spectral decompositions of adjacency matrices are proven to be consistent and asymptotically normal, as well as to satisfy a Donsker-class functional central limit theorem.  Furthermore, the power of the adjacency spectral embedding extends beyond this efficiency.  Specifically, because the adjacency spectral embedding accurately estimates the true latent positions of a latent structure model, it can be simultaneously deployed in two directions: for classical estimation of the underlying distribution $G$, whether parametric or nonparametric, as well as for manifold learning or curve-fitting of the structural support $\mathcal{C}$.
 
 We provide numerical simulations in the case of a latent structure model with underlying distributions belonging to the parametric $\textrm{Beta}(a,b)$ family on the unit interval, with structural support $\mathcal{C}$ the Hardy-Weinberg curve in the simplex. For estimating the underlying parameters $(a, b)$, we exhibit mean-squared error of a comparable order whether using parametric $M$-estimates of $(a, b)$ from the true or the estimated latent positions. Moreover, even if we do not assume full knowledge of this Hardy-Weinberg curve but merely constrain the estimate to be quadratic, we generate best-fitting Bezier curves through the point cloud of estimated latent positions, thus producing an estimate $\hat{\bC}$ for the structural support. Thereafter, we perform $M$-estimation with the points $\hat{Y}_i=p^{-1}(\pi_{\hat{\bC}}(\hat{X}_i))$ in the unit interval (recall that $\pi$ is the projection map). We show that $M$-estimation for $(a, b)$ using the points $\hat{Y}_i$, which are pullbacks of projections of $\hat{X}_i$ onto the estimated curve $\hat{\bC}$, also compares favorably with $M$-estimation for $(a, b)$ using the original, true latent positions $X_i$. We reiterate, though, that the accuracy of these classical statistical estimates for $(a, b)$ is impacted by the accuracy of the estimation for the support $\mathcal{\bC}$.  Ongoing work includes the development of a Berry-Esseen result that characterizes finite-sample performance of this $M$-estimation procedure and its dependence on the underlying distribution $G$. 
 	
 When the structural support is neither known nor parametrically specified, manifold learning procedures can be successfully exploited for subsequent inference. We demonstrate how the latent structure model provides theoretical underpinning for testing a hitherto-open neuroscientific question on bilateral homology in the right and left hemispheres of the larval Drosophila connectome. We model this connectome as a latent structure graph and focus on a specific type of neural cell, the Kenyon cell; we illustrate the use the estimated latent positions to learn the structural support for the Kenyon cells in the right and left Drosophila hemispheres. In practice, we leverage \texttt{isomap} to yield scaled, inter-point geodesic distances between estimated latent positions. We extract from this two sets of points in the unit interval, one for each hemisphere. A classical Kolmogorov-Smirnov test results in a failure to reject the null hypothesis that the right and left hemisphere subgraphs have the same underlying distributions.  However, if we reorient the estimated curve for one hemisphere but not the other, by considering $t \mapsto 1-t$ in the map $\hat{p}:[0,1] \mapsto \hat{\bC}$ for one hemisphere, we find that we {\em do} reject the null hypothesis of equality of underlying distribution.  From the point of view of geometric structure, this lack of symmetry is reassuring, and moreover it allows us to reject the hypothesis that the underlying distribution is uniform. 
 
 Current research concerns theoretical justification of efficiency when the support is unknown and must be learned, an analysis of kernel density estimation and testing for the nonparametric case, as well as the formal development of latent structure models with higher-dimensional support.  The simplicity and approximability of latent structure models is an argument for their use in representing network phenomena, and the efficiency of spectrally-derived $M$-estimates for LSM parameters is a useful dividend. The latent structure model harmonizes classical statistics, geometry, and manifold learning, and as such is an elegant platform for network inference.
 \section{Acknowledgments}\label{sec:ack}
 The authors gratefully acknowledge support from the Defense Advanced Research Programs Agency (DARPA) through the ``Data-Driven Discovery of Models" (D$3$M) Program;
 the ``Fundamental Limits of Learning" (FunLoL) Program via SIMPLEX; the Naval Engineering Education Consortium (NEEC) Office of Naval Research (ONR) Award Number N00174-19-1-0011; and the Air Force Office of Scientific Research (AFOSR) Grant FA9550-17-1-0280
 ``Foundations and Algorithms for Statistics and Learning for Data in Metric Spaces."

\bibliographystyle{plain}
\bibliography{biblio_summary}

\begin{thebibliography}{10}

\bibitem{Airoldi2008}
E.~M. Airoldi, D.~M. Blei, S.~E. Fienberg, and E.~P. Xing.
\newblock {Mixed membership stochastic blockmodels}.
\newblock {\em The Journal of Machine Learning Research}, 9:1981--2014, 2008.

\bibitem{athreya_survey}
A.~Athreya, D.~E. Fishkind, K.~Levin, , V.~Lyzinski, Y.~Park, Y.~Qin, D.~L.
  Sussman, M.~Tang, J.~T. Vogelstein, and C.~E. Priebe.
\newblock Statistical inference on random dot product graphs: a survey.
\newblock {\em Journal of Machine Learning Research}, 18, 2018.

\bibitem{athreya2013limit}
A.~Athreya, V.~Lyzinski, D.~J. Marchette, C.~E. Priebe, D.~L. Sussman, and
  M.~Tang.
\newblock A limit theorem for scaled eigenvectors of random dot product graphs.
\newblock {\em Sankhya A}, 78:1--18, 2016.

\bibitem{Bickel_Doksum}
P.~J. Bickel and K.J. Doksum.
\newblock {\em Mathematical Statistics: Basic Ideas and Selected Topics, Vol
  1}.
\newblock Pearson Prentice-Hall, second edition, 2007.

\bibitem{bollobas2007phase}
B.~Bollob{\'a}s, S.~Janson, and O.~Riordan.
\newblock The phase transition in inhomogeneous random graphs.
\newblock {\em Random Structures \& Algorithms}, 31(1):3--122, 2007.

\bibitem{breiman_statsci}
L.~Breiman.
\newblock Statistical modeling: The two cultures.
\newblock {\em Statistical Science}, 16:199--215, 2001.

\bibitem{cape2toinfty}
J.~Cape, M.~Tang, and C.~E. Priebe.
\newblock The two-to-infinity norm and singular subspace geometry with
  applications to high-dimensional statistics.
\newblock {\em Annals of Statistics}, 2018, to appear.
\newblock Arxiv preprint at \url{http://arxiv.org/abs/1705.10735}.

\bibitem{Coppersmith2014}
G.~A. Coppersmith.
\newblock Vertex nomination.
\newblock {\em Wiley Interdisciplinary Reviews: Computational Statistics},
  6:144--153, 2014.

\bibitem{diaconis2007graph}
P.~Diaconis and S.~Janson.
\newblock Graph limits and exchangeable random graphs.
\newblock {\em arXiv preprint arXiv:0712.2749}, 2007.

\bibitem{EichlerSubmitted}
K.~Eichler, F.~Li, A.~L. Kumar, Y.~Park, I.~Andrade, C.~Schneider-Mizell,
  T.~Saumweber, A.~Huser, D.~Bonnery, B.~Gerber, R.~D. Fetter, J.~W. Truman,
  C.~E. Priebe, L.~F. Abbott, A.~Thum, M.~Zlatic, and A.~Cardona.
\newblock The complete wiring diagram of a high-order learning and memory
  center, the insect mushroom body.
\newblock {\em Nature}, 548:175--182, 2017.

\bibitem{FisLyzPaoChePri2015}
D.~E. Fishkind, V.~Lyzinski, H.~Pao, L.~Chen, and C.~E. Priebe.
\newblock Vertex nomination schemes for membership prediction.
\newblock {\em Annals of Applied Statistics}, 9:1510--1532, 2015.

\bibitem{gallier}
J.-P. Gallier.
\newblock {\em Curves and Surfaces in Geometric Modeling: Theory and
  Algorithms}.
\newblock Morgan-Kaufman, 2000.

\bibitem{hoff_raftery_handcock}
P.~D. Hoff, A.~E. Raftery, and M.~S. Handcock.
\newblock Latent space approaches to social network analysis.
\newblock {\em Journal of the American Statistical Association}, 97:1090--1098,
  2002.

\bibitem{Holland1983}
P.~W. Holland, K.~Laskey, and S.~Leinhardt.
\newblock {Stochastic blockmodels: First steps}.
\newblock {\em Social Networks}, 5:109--137, 1983.

\bibitem{karrer2011stochastic}
B.~Karrer and M.~E.~J. Newman.
\newblock Stochastic blockmodels and community structure in networks.
\newblock {\em Physical Review E}, 83, 2011.

\bibitem{lee_manifold}
J.~M. Lee.
\newblock {\em Introduction to Smooth Manifolds}.
\newblock Springer-Verlag, second edition, 2013.

\bibitem{lyzinski13:_perfec}
V.~Lyzinski, D.~L. Sussman, M.~Tang, A.~Athreya, and C.~E. Priebe.
\newblock Perfect clustering for stochastic blockmodel graphs via adjacency
  spectral embedding.
\newblock {\em Electronic Journal of Statistics}, 8:2905--2922, 2014.

\bibitem{lyzinski15_HSBM}
V.~Lyzinski, M.~Tang, A.~Athreya, Y.~Park, and C.~E. Priebe.
\newblock Community detection and classification in hierarchical stochastic
  blockmodels.
\newblock {\em IEEE Transactions in Network Science and Engineering}, 4:13--26,
  2017.

\bibitem{prautzsch}
H.~Prautzsch, W.~Boehm, and M.~Paluszny.
\newblock {\em Bezier and B-spline Techniques}.
\newblock Springer, 2002.

\bibitem{MBStructure}
C.~E. Priebe, Y.~Park, M.~Tang, A.~Athreya, V.~Lyzinski, J.~T. Vogelstein,
  Y.~Qin, B.~Cocanougher, K.~Eichler, M.~Zlatic, and A.~Cardona.
\newblock Semiparametric spectral modeling of the drosophila connectome.
\newblock arXiv preprint at \url{https://arxiv.org/abs/1705.03297}, 2017.

\bibitem{asta_cls}
A.~L. Smith, D.~Asta, and C.~A. Calder.
\newblock The geometry of continuous latent space models for network data.
\newblock Arxiv preprint at \url{http://arxiv.org/abs/1712.08641}, 2017.

\bibitem{STFP-2011}
D.~L. Sussman, M.~Tang, D.~E. Fishkind, and C.~E. Priebe.
\newblock A consistent adjacency spectral embedding for stochastic blockmodel
  graphs.
\newblock {\em Journal of the American Statistical Association},
  107:1119--1128, 2012.

\bibitem{tang14:_nonpar}
M.~Tang, A.~Athreya, D.~L. Sussman, V.~Lyzinski, and C.~E. Priebe.
\newblock A nonparametric two-sample hypothesis testing problem for random dot
  product graphs.
\newblock {\em Bernoulli}, 23:1599--1630, 2017.

\bibitem{tang_sbm_eff}
M.~Tang, J.~Cape, and C.~E. Priebe.
\newblock Asymptotically efficient estimators for stochastic blockmodels: the
  naive mle, the rank-constrained mle, and the spectral.
\newblock Arxiv preprint at \url{https://arxiv.org/pdf/1710.10936.pdf}, 2017.

\bibitem{isomap_science}
J.~B. Tenenbaum, V.~de~Silva, and J.C. Langford.
\newblock A global geometric framework for nonlinear dimensionality reduction.
\newblock {\em Science}, 290:2319---2323, 2000.

\bibitem{vaart96:_weak}
A.~W. van~der Vaart and J.~A. Wellner.
\newblock {\em Weak convergence and empirical processes: with applications to
  statistics}.
\newblock Springer, 1996.

\bibitem{young2007random}
S.~Young and E.~Scheinerman.
\newblock Random dot product graph models for social networks.
\newblock In {\em Proceedings of the 5th international conference on algorithms
  and models for the web-graph}, pages 138--149, 2007.

\end{thebibliography}
\end{document}